\documentclass{amsart}

\def\bege{\begin{equation}} \def\ende{\end{equation}}
\def\begr{\begin{eqnarray}} \def\endr{\end{eqnarray}}
\usepackage{cite}

\usepackage{bbm}
\usepackage{amsmath}
\usepackage{txfonts}
\usepackage{stmaryrd}
\usepackage{amssymb}
\usepackage{amsfonts}
\usepackage{times}
\usepackage{mathrsfs}
\usepackage{comment}

\usepackage{amsthm}
\usepackage{enumerate}

\usepackage{framed}
\usepackage{color}

\newcommand{\B}{{\mathbb{B}_n}}

\newcommand{\D}{{\mathbb D}}
\newcommand{\C}{{\mathbb C}}

\newcommand{\N}{{\mathbb N}}
\newcommand{\A}{{ \mathcal A}}

\def \psi {u}
\def\a{\alpha}

\def\e{ \varepsilon}

\def\begr{\begin{eqnarray}} \def\endr{\end{eqnarray}}

\DeclareMathOperator*{\esssup}{ess\,sup}

\newtheorem{remark}{Remark}

 \newtheorem{thm}{Theorem}[section]
 \newtheorem{cor}[thm]{Corollary}
 \newtheorem{lem}[thm]{Lemma}
 
 \theoremstyle{definition}
 
 \theoremstyle{remark}
 \newtheorem{rem}[thm]{Remark}
 
 \numberwithin{equation}{section}

\def\bege{\begin{equation}} \def\ende{\end{equation}}
\def\ve{\varepsilon} 
\def\qand{\quad\mbox{ and }\quad}

\def\bege{\begin{equation}} \def\ende{\end{equation}}
\def\begr{\begin{eqnarray}} \def\endr{\end{eqnarray}}
\def\bnum{\begin{enumerate}} \def\enum{\end{enumerate}}

\begin{document}

\title[ Tent space theory and Derivative Area operators]
 {Tent space theory and Derivative Area operators of Hardy spaces into Lebesgue spaces}

\author[Xiaosong Liu]{Xiaosong Liu}
\address{%
    Xiaosong Liu:  Department of Mathematics, Jiaying University, Meizhou 514015, China;}
\email{gdxsliu@163.com}
\thanks{The first author is in part supported by  Department of Education of Guangdong Province (No. 2024KTSCX084).
The second author is in part supported by National Natural Science Foundation of China (No. 12471093), Guangdong Basic and Applied Basic Research Foundation (No. 2024A1515010468), and LKSF STU-GTIIT Joint-research Grant (No. 2024LKSF06).}

\author{Zengjian Lou$^\dagger$}\thanks{$^\dagger$Corresponding author.}
\address{%
    Zengjian Lou:  Department of Mathematics, Shantou University, Shantou 515063, China;}
\email{zjlou@stu.edu.cn}

\author{Zixing Yuan}
\address{%
     Zixing Yuan: School of Mathematics and Statistics, Wuhan University, Wuhan 430072, China;}
\email{zxyuan.math@whu.edu.cn}

\author{Ruhan Zhao$^\dagger$}
\address{Ruhan Zhao:  Department of Mathematics, SUNY Brockport, Brockport, NY 14420, USA.}
\email{rzhao@brockport.edu}

\subjclass[2000]{Primary 47B38, Secondary 32A35, 32A37}

\keywords{Hardy space, Tent space, Derivative area operator, Carleson measure}

\date{\today}

\begin{abstract}
The tent space theory is an effective tool for studying operators from Hardy spaces $H^p$ to Lebesgue spaces $L^q$ for $0<q<p<\infty$. 
The theory of tent spaces  was established by Coifman, Mayer and Stein in 1985.
In this paper, we completely characterize the boundedness and compactness of derivative area operators 
from  $H^p$ on the unit ball $\B$ into $L^q(\mathbb{S}_n)$ for $0<p,q<\infty$. 
As byproducts, we also give complete characterizations for a positive Borel measure $\mu$ 
such that $\mathcal{R}^k:\,H^p\to L^q(\mu)$ is bounded or compact on the unit ball $\B$ for all $0<p,q<\infty$ in terms of Carleson measure and tent spaces,
where $\mathcal{R}^k$ ($k\in\N$) is the radial derivative of order $k$.
\end{abstract}

\maketitle

\section{Introduction}

The area operator induced by a positive measure $\mu$
on the unit disk $\D$ was first studied by Cohn \cite{WSC1997} in 1997,
in which it is showed that the area operator is bounded from 
the Hardy space $H^p(\D)$ to the Lebesgue space $L^p(\partial\D)$
if and only if $\mu$ is a Carleson measure. 
Here $\partial\D$ is the unit circle.
The area operator is considered as a generalization of the area integral, 
which is closely related to Hardy spaces.
In this paper we investigate derivative area operators, which are
compositions of area operators and differentiation operators, 
on the unit ball $\B$ in $\C^n$.

For any two points $z=(z_ 1,\dots,z_ n)$ and $w=(w_ 1,\dots,w_ n)$ in $\C^n$
we write
$\langle z,w\rangle =z_ 1\bar{w}_ 1+\dots +z_ n \bar{w}_ n,$
and
$|z|=\sqrt{\langle z,z\rangle}=\sqrt{|z_ 1|^2+\dots +|z_ n|^2}.$
Let $\mathbb{B}_n=\{z\in \C^n:|z|<1\}$ be the unit ball in $\C^n$,
and let $\mathbb{S}_n=\{z\in \C^n:|z|=1\}$ be the unit sphere. 
We use $dV$ and $d\sigma$ to denote the normalized volume measure on $\mathbb{B}_n$
such that $V(\mathbb{B}_n)=1$
and the normalized area measure on $\mathbb{S}_n$ 
such that $\sigma(\mathbb{S}_n)=1$, respectively,
and denote by $d\tau(z)=dV(z)/(1-|z|^2)^{n+1}$.

For any $\zeta\in \mathbb{S}_n$ and $\alpha>1$, the admissible approach region in $\mathbb{B}_n$ is defined by
\begin{equation}\label{Gamma}
\Gamma_\alpha(\zeta)
=\{z\in\mathbb{B}_n: |1-\langle z,\zeta\rangle|<\alpha(1-|z|)\}.
\end{equation}
In fact, the choice of $\a>1$ (called the aperture) is not important for our purpose 
(see \cite{CMS1985} and \cite{PR2015}), and so
we sometimes drop $\a$ and write it as $\Gamma(\zeta)$. 
We will also write $\widetilde{\Gamma}(\zeta)$ and 
$\widetilde{\widetilde{\Gamma}}(\zeta)$ to indicate 
admissible approach regions with different apertures.

Let $H(\mathbb{B}_n)$ be the space of all holomorphic functions on $\mathbb{B}_n$. 
For $0<p<\infty$, the Hardy space $H^p$ consists of
those  functions $f\in H(\mathbb{B}_n)$ such that
\[
\|f\|_{H^p}^p=\sup_{0<r<1}\int_{\mathbb{S}_n}|f(r\zeta)|^p\,d\sigma(\zeta)<\infty.
\]
We refer to the books \cite{AA1994}, \cite{WR1980} and \cite{KZ2005} for the
theory of Hardy spaces on the unit ball.

For a multi-index $m=(m_1,\cdots,m_n)$ and $f\in H(\B)$, 
let $|m|=m_1+m_2+\cdots+m_n$,
$z^m=z_1^{m_1}\cdots z_n^{m_n}$,
$\partial f(z)=\nabla f(z)=(\partial f/\partial z_1,\cdots,\partial f/\partial z_n)$
and $\partial^mf=\partial^{|m|}f/(\partial z_1^{m_1}\cdots \partial z_n^{m_n})$. 

Suppose the homogeneous expansion of $f\in H(\B)$ is given by
$$
f(z)=\sum_{j=0}^{\infty}f_j(z).
$$ 
Let $\alpha$ is an real parameter. 
we denote a class of fractional radial derivative operator 
$\mathcal{R}^\alpha: H(\B)\rightarrow  H(\B)$ as follows:
\[
\mathcal{R}^\alpha f(z)=\sum_{j=1}^{\infty}j\,^{\alpha} f_j(z).
\]
Throughout this paper we will always assume that $\mu$ is a positive Borel measure on $\B$, 
finite on compact subsets of $\B$.
Many authors studied the boundedness of the embedding derivative operator 
$\partial^m:H^p \rightarrow L^p(\mu)$, i.e.
\begr\label{qhpcm}
\left(\int_{\mathbb{B}_n}|\partial^m f(z)|^q\,d\mu(z)\right)^{1/q}\leq C\|f\|_{H^p}.
\endr
See for example, Luecking \cite{HL1991}, Jevti\'{c} \cite{JM1995}, \cite{JM1996} and Arsenovi\'{c} \cite{MA1999}.
Let $k$ be a positive integer. 
It is known that embedding derivative operators $\mathcal{R}^k:H^p \rightarrow L^p(\mu)$ 
have similar characterizations to embedding derivative operators $\partial^m:H^p \rightarrow L^p(\mu)$.
See, for example \cite{WZ}.
In this paper we consider the following derivative area operator on the Hardy space $H^p$:
\[
\A_{\mu,\,s}^k(f)(\zeta)=\left(\int_{\Gamma(\zeta)}|\mathcal{R}^kf(z)|^s\frac{d\mu(z)}{(1-|z|)^n}\right)^{1/s},
\qquad \textrm{ for all }\,\,\zeta\in  \mathbb{S}_n.
\]
When $s=1$ and $k=0$, we simply denote $\A_{\mu,\,s}^k$ by $\A_{\mu}$.
Minkowski's inequality shows that $\A_{\mu,\,s}^k$ is sublinear.
 
The main goal of this paper is to characterize
boundedness and compactness of derivative area operators $\A_{\mu,\,s}^k$ from $H^p$ into $L^q(\mathbb{S}_n)$
for $0<p,\,q<\infty$.
Relating to this question, 
Cohn \cite{WSC1997} studied $\A_\mu$ from $H^p(\D)$ to $L^p(\partial\mathbb{D})$ on the unit disk $\D$. 
Gong, Lou and Wu \cite{GLW2010} studied $\A_\mu$ from $H^p(\D)$ to $L^q(\partial\mathbb{D})$ when $0<p,\,q<\infty$. 
These results have been generalized to the unit ball in \cite{LLZ2022} and \cite{LP2024}.
The boundedness and compactness of $\A_\mu$ on weighted Bergman spaces on the unit disk $\mathbb{D}$ 
have been also studied intensively recently 
(see, e.g., \cite{ATTY2024}, \cite{PRS2015}, \cite{ZW2006}). 

For $z\in \B$ and $r>0$, let $D(z,r)=\{w\in \B:\beta(w,z)<r\}$ denote the Bergman metric ball 
centered at a point $z\in \B$ with radius $r$, 
where $\beta$ be the Bergman metric on $\B$. 
The following is our main result.
The result involves tent spaces $T^p_q(\tau)$ with $0<p<\infty$, $0<q\le\infty$,
whose definitions will be given in Section 3.

\begin{thm}\label{thm1.1}
Let $k$ be a positive integer, let $0<p, q ,s<\infty$ and let $0<r<1$. 
Let $\mu$ be a positive Borel measure on $\mathbb{B}_n$ that is finite
on compact subsets of $\mathbb{B}_n$. 
Let $\Phi_{\mu}(z)=\Phi_{\mu,r}(z)=\mu(D(z,r))/(1-|z|)^{n+ks}$.
Then the following statements hold.
\begin{itemize}
\item[(i)] If $p<q,\,s>0$ or $p=q, s\geq 2$, then
$\A_{\,\mu,\,s}^k:H^p\rightarrow L^q(\mathbb{S}_n)$
is bounded if and only if $\mu$ is a $(ks/n+1+s/p-s/q)$-Carleson measure,
or equivalently,
$$
\sup_{z\in\B}\frac{\mu(D(z,r))}{(1-|z|)^{ks+(1+s/p-s/q)n}}<\infty.
$$
\item[(ii)] If $p=q$ and $s<2$, then
$\mathcal{A}^k_{\mu,s}: H^p \to L^q(\mathbb S_n)$ is bounded if and only if 
$\Phi_{\mu}\in T^{\infty}_{2/(2-s)}(\tau)$.
\item[(iii)] If  $p>q$ and $s \geq 2$, then $\mathcal{A}_{\mu, s}^k: H^p \to L^q(\mathbb{S}_n)$
 is bounded if and only if $\Phi_{\mu} \in T_{\infty}^{pq/(s(p-q))}(\tau)$.
\item[(iv)] If $p>q$ and  $s <2$, then $\mathcal{A}_{\mu,s}^k: H^p \to L^q(\mathbb{S}_n)$ is bounded 
if and only if $\Phi_{\mu}\in T_{2/(2-s)}^{pq/(s(p-q))}(\tau)$.
\end{itemize}
\end{thm}

In the proof of the previous theorem, we can
always assume $f(0)=0$, because for any constant c, $\A_{\,\mu,\,s}^kf=\A_{\,\mu,\,s}^k(f+c)$.
We will also give a complete characterization for compactness of $\mathcal{A}^k_{\mu,s}: H^p \to L^q(\mathbb S_n)$.
As byproducts, we will give complete characterizations for a positive Borel measure $\mu$ 
such that $\mathcal{R}^k:H^p\to L^q(\mu)$ is bounded or compact on the unit ball $\B$ for $0<p,q<\infty$.

The paper is organized as follows.
In Section 2 we give  preliminary results.
In Section 3 we give some characterizations of tent spaces on $\mathbb{B}_n$.
Proof of Theorem 1.1 is in Section 4.
In Section 5, we characterize the compactness of $\mathcal{A}^k_{\mu,s}: H^p \to L^q(\mathbb S_n)$.

Throughout this paper, constants are denoted by $C$.
They are positive and may differ from one occurrence to the other.
The notation $a \lesssim b$ means that there is a positive constant $C$
such that $a \leq C b$.
Moreover, $a \asymp b$ means both $a\lesssim b$ and $b\lesssim a$ hold.
For  $p\geq 1$, we denote its conjugate by $p'$.
Thus $1/p+1/{p'}=1$ for $1<p<\infty$ and $p'=\infty$ for $p=1$.

\section{Preliminary results}
In this section, we give some auxiliary results which will be used in the rest of this paper.

For $w\in \mathbb{C}^n$ and $t>0$, we denote by $B(w,t)$ the Euclidean ball centered at $w$ with radius $t$. 
For $z,w\in\overline{\mathbb{B}}_n$, let
$$
d(z,w)=|1-\langle z,w\rangle|^{1/2}.
$$
It is known that, for any points $z,u,w\in \overline{\mathbb{B}}_n$,
\begin{equation}\label{trig}
d(z,w)\le d(z,u)+d(u,w),
\end{equation}
and, if we restrict $d$ on $\mathbb{S}_n$ it becomes a metric. 
See, for example, \cite[page 113]{KZ2005}.

For $\zeta\in \mathbb{S}_n$ and $r>0$, let
$$
Q(\zeta,r)=\{\xi\in \mathbb{S}_n:|1-\langle \zeta,\xi\rangle|^{1/2}<r\}
$$
and
$$
B_r(\zeta)=\{z\in \mathbb{B}_n:|1-\langle z,\zeta\rangle|^{1/2}<r\}.
$$
$Q(\zeta,r)$ is a non-isotropic metric ball on $\mathbb{S}_n$ centered at $\zeta$ 
with radius $r$.
We call $Q(\zeta,r)$ a $d$-ball following \cite{KZ2005},
and $B_r(\zeta)$ a Carleson tube.
It is obvious that, for $r>\sqrt{2}$, we have $Q(\zeta,r)=\mathbb{S}_n$ and $B_r(\zeta)=\B$.
In particular, for any $a\in \mathbb{B}_n\setminus\{0\}$, we denote by
\[
Q_a=Q\left(\frac{a}{|a|},\sqrt{1-|a|}\right)
=\left\{\zeta\in \mathbb{S}_n:\left|1-\left\langle \zeta,\frac{a}{|a|}\right\rangle\right|<1-|a|\right\}
\]
and
\[
B_a=B_{\sqrt{1-|a|}}\left(\frac{a}{|a|}\right)
=\left\{z\in \mathbb{B}_n:\left|1-\left\langle z,\frac{a}{|a|}\right\rangle\right|<1-|a|\right\}.
\]
If $a=0$, we define $B_a=\mathbb{B}_n$ and $Q_a=\mathbb{S}_n$.
By Lemma 4.6 in \cite{KZ2005}, we know that
\begin{equation}\label{qr}
\sigma(Q(\zeta,r))\asymp r^{2n}
\end{equation}
for a fixed $r$ with $0<r<\sqrt{2}$ and for any $\zeta\in\mathbb{S}_n$,
and so $\sigma(Q_a)\asymp(1-|a|)^n$.

For a $d$-ball $Q(\zeta,r)$ in $\mathbb{S}_n$, we also define 
$$
S(Q(\zeta,r))=\left\{z\in\mathbb{B}_n:\,\frac{z}{|z|}\in Q(\zeta,r), 1-\frac{r^2}{2}<|z|<1\right\}. 
$$
Again, if $r>\sqrt{2}$, we have $S(Q(\zeta,r))=\B$.
For convenience, we denote by $S(a)=S(Q_a)$. 
Let  $0<r<1$, $\zeta\in\mathbb{S}_n$.
Denote by $a=(1-r^2)\zeta$. Then $a/|a|=\zeta$ and $r=\sqrt{1-|a|}$.
Hence,
\begin{equation}\label{BQ}
B_r(\zeta)=B_a \qand Q(\zeta,r)=Q_a.
\end{equation}

For $\alpha>1$ and $z\in\B$, let 
\[
I(z)=I_{\alpha}(z)=\{\zeta\in \mathbb{S}_n:\,z\in \Gamma_{\alpha}(\zeta)\},
\]
where $\Gamma_{\alpha}(\zeta)$ is given by (\ref{Gamma}).
By (\ref{trig}) we can easily see that
\begin{equation}\label{IQ}
Q\left(\frac{z}{|z|}, \left(\sqrt{\frac{\a}{2}}-1\right)\sqrt{1-|z|}\right)
\subseteq I_{\a}(z)
\subseteq Q\left(\frac{z}{|z|}, \left(1+\sqrt{\frac{\a}{2}}\right)\sqrt{1-|z|}\right)
\end{equation}
for $\alpha>2$. Hence,
\begr\label{ei(z)}
\sigma\left(I_{\alpha}(z)\right)
\asymp \sigma(Q_z)
\asymp(1-|z|)^n.
\endr

Let $F$ be any closed subset of $\mathbb{S}_n$ 
and $O=\mathbb{S}_n\setminus F$ be the complement of $F$ on $\mathbb{S}_n$.
Let $\alpha>1$.
We denote by $\mathcal{R}_{\alpha}(F)=\bigcup_{\zeta\in F}\Gamma_{\alpha}(\zeta)$
and $\Lambda(O)=\Lambda_{\alpha(O)}=\mathbb{B}_n\backslash \mathcal{R}_{\alpha}(F)$. 
The set $\Lambda(O)$ is called the tent over $O$.
If $F=\emptyset$ then $O=\mathbb{S}_n$. 
In this case we define $\Lambda(\mathbb{S}_n)=\mathbb{B}_n$. 
For a $d$-ball $Q$ in $\mathbb{S}_n$, the tent over $Q$ 
can be equivalently defined as follows:
$$
\Lambda_{\alpha}(Q)=\{w\in \mathbb{B}_n: I_{\alpha}(w)\subset Q\}.
$$
We will denote $\Lambda(I_{\alpha}(z))$ by $\Lambda_{\alpha}(z)$ or simply $\Lambda(z)$.

Let $s>0$.
We say that a positive Borel measure $\mu$ on $\B$ is a (bounded) $s$-Carleson measure 
in $\mathbb{B}_n$ if
\begin{equation}\label{cm}
\sup_{\zeta\in\mathbb{S}_n,r>0}\frac{\mu(S(Q(\zeta,r)))}{\left(\sigma(Q(\zeta,r))\right)^s}<\infty,
\end{equation}
and $\mu$ is a vanishing $s$-Carleson measure in $\mathbb{B}_n$ if
\begin{equation}\label{vcm}
\lim_{r\to0}\frac{\mu\left(S(Q(\zeta,r))\right)}{\left(\sigma(Q(\zeta,r))\right)^s}=0
\end{equation}
uniformly in $\zeta$.
When $s=1$, we will simply call these measures
the Carleson measure and the vanishing Carleson measure, respectively.

For $0<r<1$ and $\zeta\in \mathbb{S}_n$, we can easily get that
\begin{equation}\label{ect}
S(Q(\zeta,r/2))\subseteq B_r(\zeta)\subseteq S(Q(\zeta,2r)),
\end{equation}
and
\begin{equation}\label{ett}
S(Q(\zeta,r))\subseteq \Lambda(Q(\zeta,3r))\subseteq S(Q(\zeta,6r)).
\end{equation}
Thus, for $a\in\B$, by (\ref{qr}), (\ref{ect}) and (\ref{ett}) we have
\begin{eqnarray}\label{evctt}
V(B_a)&\asymp& V(\Lambda(a))\asymp V(S(a))\\
&\asymp&\int_{Q_a}\int_{(1+|a|)/2}^12nr^{2n-1}drd\sigma(\zeta)\asymp (1-|a|)^{n+1}.\nonumber
\end{eqnarray}
and in the above definitions of $s$-Carleson measures
we may replace $S(Q(\zeta,r))$ by $\Lambda(Q(\zeta,r))$ or $B_r(\zeta)$, and replace $\sigma(Q)$ by $r^{2n}$.
We may also rewrite  them in the following forms:
$\mu$ is an $s$-Carleson measure in $\mathbb{B}_n$ if
\begin{equation}\label{cm1}
\sup_{a\in\mathbb{S}_n}\frac{\mu(S(a))}{(1-|a|)^{ns}}
<\infty,
\end{equation}
and $\mu$ is a vanishing $s$-Carleson measure in $\mathbb{B}_n$ if
\begin{equation}\label{vcm1}
\lim_{|a|\to1}\frac{\mu(S(a))}{(1-|a|)^{ns}}=0,
\end{equation}
and we may replace $S(a)$ in the above definitions by $\Lambda(a)$ or $B_a$.
It is also well-known that, if $s>1$, for a fixed $r\in(0,1)$, 
we may replace $S(a)$ by $D(a,r)$ in (\ref{cm1}) and (\ref{vcm1}).  

We need some well-know characterizations of Carleson measures and  vanishing Carleson measures.
If $\mu$ is a Carleosn measure on $\B$, we denote by
\begin{equation}\label{c-norm}
\|\mu\|=\sup_{a\in\mathbb{S}_n}\frac{\mu(S(a))}{(1-|a|)^{n}}.
\end{equation}
The following result for the unit disk is due to Carleson, and for the unit ball can be found
in Theorem 5.4 in \cite{KZ2005}.

\begin{lem}\label{carleson} Let $\mu$ be a positive Borel measure on $\B$,
finite on compact subsets of $\B$.
Then $\mu$ is a Carelson measure on $\B$ if and only if
$$
\sup_{a\in\B}\int_{\B}\frac{(1-|a|)^n}{|1-\langle a,z\rangle|^{2n}}\,dV(z)<\infty.
$$
Moverover, 
$$
\|\mu\|\asymp\sup_{a\in\B}\int_{\B}\frac{(1-|a|)^n}{|1-\langle a,z\rangle|^{2n}}\,dV(z).
$$
\end{lem}

\begin{lem}\label{carleson-v} Let $\mu$ be a positive Borel measure on $\B$.
Let $0<R<1$, and let $\mu_R=\mu|_{\B\setminus\overline{B(0,R)}}$.
Suppose that $\mu$ is a Carleson measure on $\B$. Then the following statements are 
equivalent.
\begin{itemize}
\item[(i)] $\mu$ is a vanishing Carelson measure on $\B$.
\item[(ii)] 
$$
\lim_{|a|\to1}\int_{\B}\frac{(1-|a|)^n}{|1-\langle a,z\rangle|^{2n}}\,d\mu(z)=0.
$$
\item[(iii)] 
$$
\lim_{R\to1}\|\mu_R\|=0.
$$
\end{itemize}
\end{lem}

\begin{proof} The proof for the equivalence between (i) and (ii) can be found, for example, 
the proof of Theorem 5.10 in \cite{KZ2005}.
The proof for (i)$\Rightarrow$(iii) can be found in page 130 in \cite{CM1995}.
The proof for (iii)$\Rightarrow$(ii) for the case of the unit disk was given in \cite{LL2015}.
For the unit ball the proof is similar. We give it here for completion.
Thus 
\begin{eqnarray*}
\int_{\B}\frac{(1-|a|)^n}{|1-\langle a,z\rangle|^{2n}}\,d\mu(z)
&=&\int_{\B}\frac{(1-|a|)^n}{|1-\langle a,z\rangle|^{2n}}\,d\mu_R(z)
 +\int_{\B}\frac{(1-|a|)^n}{|1-\langle a,z\rangle|^{2n}}\,d(\mu-\mu_R)(z)\\
&\lesssim& \left\|\mu_R\right\|+\int_{\B}\frac{(1-|a|)^n}{|1-\langle a,z\rangle|^{2n}}\,d(\mu-\mu_R)(z)
\end{eqnarray*}
Since $\mu$ is a Carleson measure, 
we can easily see that $\mu-\mu|_R=\mu|_{\overline{B(0,R)}}$ is a vanishing Carleson measure.
Letting $|a|\to1$ and then $R\to 1$ in the above inequality, we get that
$$
\lim_{|a|\to1}
\int_{\B}\frac{(1-|a|)^n}{|1-\langle a,z\rangle|^{2n}}\,d\mu(z)=0.
$$
Thus $\mu$ is a vanishing Carlesn measure. The proof is complete.
\end{proof}

We will use the classical Khinchine's inequality. 
For all $k \in \mathbb N = \{1,2, \ldots\}$ we denote
by $r_k :[0,1] \to \{0,\pm 1\}$, $r_k(t) = {\rm sign} \, \sin (2^k \pi t )$, the $k$th Rademacher function. 

\medskip

\noindent \emph{Khinchine's inequality}: Let $0<p<\infty$. Then, 
\begin{align}\label{001}
\left(\sum_{k}|c_k|^{2}\right)^{p/2} 
\asymp \int_{0}^{1}\left|\sum_{k} c_{k} r_{k}(t)\right|^{p}\,dt,
\end{align}
where $\{c_k\}_{k=1}^\infty$ is an arbitrary sequence of complex scalars. 
See Appendix A in \cite{PD2000} for more details.
\medskip

Next, we recall Kahane's inequality. The details can be found in \cite{HL1993}. 

\medskip

\noindent
\emph{Kahane's inequality}: Let $X$ be a quasi-Banach space and let $0<p, q<\infty$. 
There holds 
\begin{align}\label{002}
\left(\int_{0}^{1}\left\|\sum_{k} r_{k}(t) x_{k}\right\|_{X}^{q} 
dt \right)^{1 / q} \asymp\left(\int_{0}^{1}\left\|\sum_{k} r_{k}(t) x_{k}
\right\|_{X}^{p} d t\right)^{1 / p},
\end{align}
where $\left\{x_{k}\right\}_{k=1}^\infty$ is an arbitrary sequence in $X$. 
Moreover, the implicit constants in \eqref{002}  depend only on $p$ and $q$ and not on the quasi-Banach space $X$.

We need several lemmas. 
Recall that an admissible maximal function for a continuous function $f$ on $\mathbb{B}_n$ is defined by
\[
f^*(\zeta)=\sup_{z\in \Gamma(\zeta)}|f(z)|.
\]
The following result can be found in \cite[Theorem 5.6.5]{WR1980}.

\begin{lem}\label{l2.41}
Let $0<p<\infty$ and $f\in H(\mathbb{B}_n)$. Then
\[
\|f^*\|_{L^p(\mathbb{S}_n)}\leq C \|f\|_{H^p}.
\]
\end{lem}

The next result may be known for experts. Its proof for the case $n=1$
can be found in \cite[Lemma 2.1]{HL1985}. However, we could not find the
exact reference for the case $n>1$. 
Since it is one of the key lemmas in our paper, we provide a proof here for completion.

\begin{lem}\label{dmk}
Let $k \geq m$ be non-negative integers, $0<p<\infty$ and $0<r<1$. Then for $f\in H(\B)$ and $w\in \B$, we have 
\begin{equation}\label{rf}
|\mathcal{R}^kf(w)|^p\lesssim (1-|w|)^{(m-k)p}\int_{D(w,r)}|\mathcal{R}^mf(z)|^p\,d\tau(z),
\end{equation}
where the inequality constant depends on $r$.
\end{lem}

\begin{proof} Clearly, we just need to prove (\ref{rf}) for the case $m=0$.
The case $k=0$ is a direct consequence of Lemma 2.24 and Lemma 2.20 in \cite{KZ2005}.
Now we consider the case $k=1$. 
Let $0<p<\infty$, and let $0<r<1$.
By Lemma 2.4 in \cite{KZ2005} we can easily obtain that, for any $f\in H(\B)$,
$$
|\nabla f(0)|^p\lesssim \int_{D(0,r)}|f(z)|^p\,dV(z)
$$
(see page 66 in \cite{KZ2005} for details).
Replacing $f$ by $f\circ \varphi_a$, where $\varphi_a$ is the automorphism of $\B$
that interchanges $0$ and $a$ (see, for example, p. 5 in \cite{KZ2005}), 
and by a change of variable, we get
$$
|\nabla f\circ\varphi_a(0)|^p
\lesssim \int_{D(0,r)}|f\circ\varphi_a(z)|^p\,dV(z)
=\int_{D(a,r)}|f(w)|^p\left(\frac{1-|a|^2}{|1-\langle w,a\rangle|^2}\right)^{n+1}\,dV(w)
$$
By Lemma 2.14 and Lemma 2.20 in \cite{KZ2005},  we get
\begin{eqnarray*}
(1-|a|)^p|\mathcal{R}f(a)|^p
\lesssim |\nabla f\circ\varphi_a(0)|^p
\lesssim\int_{D(a,r)}|f(w)|^p\,d\tau(w),
\end{eqnarray*}
which proved (\ref{rf}) for $k=1$.
For $k>1$ we use induction. Suppose (\ref{rf}) holds for $k\ge 1$, we prove it for $k+1$.
By the case $k=1$ and the induction assumption, we get
\begin{eqnarray*}
|\mathcal{R}^{k+1}f(a)|^p
&=&|\mathcal{R}(\mathcal{R}^{k}f)(a)|^p\\
&\lesssim&(1-|a|)^{-p}\int_{D(a,r)}|\mathcal{R}^kf(w)|^p\,d\tau(w)\\
&\lesssim&(1-|a|)^{-p}\int_{D(a,r)}(1-|w|)^{-kp}\int_{D(w,r)}|f(z)|^p\,d\tau(z)\,d\tau(w)
\end{eqnarray*}
By Lemma 2.20 in \cite{KZ2005} and Fubini's theorem, we obtain
\begin{eqnarray*}
|\mathcal{R}^{k+1}f(a)|^p
&\lesssim&(1-|a|)^{-(k+1)p}\int_{D(a,2r)}|f(z)|^p\int_{D(z,r)}\,d\tau(w)\,d\tau(z)\\
&\asymp&(1-|a|)^{-(k+1)p}\int_{D(a,2r)}|f(z)|^p\,d\tau(z).
\end{eqnarray*}
The induction proof is complete.
\end{proof}

The following result can be found in \cite[Lemma 2.3]{ZW2006}.

\begin{lem}\label{l51}
Suppose $\sigma>0,r>0$ and $\zeta \in \mathbb S_n$. Then there exists a $\sigma_r>0$ such that
\[
D(z,r)\subset \Gamma_{\sigma_r}(\zeta),\qquad\forall z\in \Gamma_\sigma(u).
\]
\end{lem}

The next result is from \cite[Lemma 1]{PD2000}.

\begin{lem}\label{duren}
Let $0<p<\infty$. If
$\{\varphi_i\}_{n=0}^\infty\subset L^p(\mathbb{S}_n)$ and $\varphi\in L^p(\mathbb{S}_n)$ satisfy
$\displaystyle\lim_{i\rightarrow \infty}\|\varphi_i\|_{L^p(\mathbb{S}_n)}=\|\varphi\|_{L^p(\mathbb{S}_n)}$ and
$\displaystyle\lim_{i\to\infty}\varphi_i(\zeta)=\varphi(\zeta)$ a.e. on $\mathbb{S}_n$,
then $\displaystyle\lim_{i\to\infty}\|\varphi_i-\varphi\|_{L^p(\mathbb{S}_n)}=0$.
\end{lem}

The next result is Lemma 4 in \cite{MA1999}.

\begin{lem}\label{arsen}
Let $\mu$ be a positive Borel measure in $\mathbb{B}_n$,
let $s>0$ and let $t>max(1,1/s)$.
Then 
\begin{eqnarray*}
\int_{\mathbb{S}_n}\left(\int_{\mathbb{B}_n}\left(\frac{1-|z|}
{|1-\langle\zeta, z\rangle|}\right)^{nt} d\mu(z)\right)^s\,d\sigma(\zeta)
\asymp
\int_{\mathbb{S}_n}\left(\mu\left(\Gamma(\zeta)\right)\right)^s\,d\sigma(\zeta).
\end{eqnarray*}
\end{lem}

The following result is well-known. 
See, for example, Problem 18 on p. 74 of \cite{WR1987}.

\begin{lem}\label{rudin}
Let $1\le p<\infty$, let $L^p(X,\mu)$ be the $L^p$ space on a measure space $X$ with a positive Borel measure,
and let $\{f_i\}$ be a sequence in $L^p(X,\mu)$ such that 
$\|f_i-f\|_{L^p(X,\mu)}\to0$ as $i\to\infty$. Then $\{f_i\}$ has a subsequence $\{f_{i_m}\}$
such that $f_{i_m}\to f$ almost everywhere on $X$.
\end{lem}

\begin{rem}\label{rudin1} Note that, if $\{f_i\}$ is a nonnegative sequence, 
then the above result is also true for $0<p<1$ since in this case,
$f\in L^p(X,\mu)$ if and only if $f^p\in L^1(X,\mu)$.
\end{rem}




%

We also need a well-known covering lemma.
A sequence $\{a_j\}$ in $\B$ is said to be 
 \textit{separated } (in the Bergman metric)
if there is a constant $\delta>0$ such that
$\beta(a_i,a_j)>\delta$ for all $i\neq j$. 
The following result is from Theorem 2.23 in \cite{KZ2005}.

\begin{lem}\label{covering}
There exists a positive integer $N$ such that, for any $0<r<1$, 
there is a sequence $\{a_j\}$ in $\B$ satisfying the following properties:
\begin{itemize}
\item[(i)] $\B=\cup_{j}D(a_j,r)$.
\item[(ii)] The sets $D(a_j,r/4)$ are mutually disjoint.
\item[(iii)] Each point $z\in\B$ belongs to at most $N$ of the sets $D(a_j,4r)$.
\end{itemize}
\end{lem}

If a sequence  $\{a_j\}$ in $\B$ satisfies the conditions of the above lemma,
it is called a \textit{lattice} (or an $r$-\textit{lattice}).

\section{Tent spaces}

The concept of tent spaces was introduced by Coifman, Meyer and Sein \cite{CMS1985} in 1985. 
The theory of tent space is an useful tool to study the operator theory on Hardy spaces and weighted Bergman spaces. 
See Pau \cite{JP2016}, Liu, Lou and Zhao \cite{LLZ2019}, Lv and Pau \cite{LP2024}, Miihkinen, Pau, Per\" al\"a and Wang \cite{Mi}
for some applications of tent spaces.

Let $\omega$ be a positive Borel measure on $\mathbb{B}_n$, finite on compact subsets.
Let $0<q<\infty$. Denote by
\[
A_{q,\,\omega}^q(f)(\zeta)=\int_{\Gamma(\zeta)}|f(z)|^q\,d\omega(z)
\]
and
\[
A_{\infty,\,\omega}(f)(\zeta)=\omega-\esssup_{z\in\Gamma(\zeta)}|f(z)|,
\]
where $\zeta\in\mathbb S_n$.
For $0<p<\infty$ and $0<q\leq \infty$, the tent space $T_q^p(\omega)$ is the space of all
measurable functions $f$ on $\mathbb B_n$ such that
\[
\|f\|_{T_q^p(\omega)}=\|A_{q,\,\omega}(f)\|_{L^p(\mathbb{S}_n)}<\infty.
\]

For $0<q<\infty$ and $\zeta\in\mathbb S_n$, denote by
\[
C_{q,\,\omega}^q(f)(\zeta)
=\sup_{a\in\Gamma(\zeta)}\frac{1}{(1-|a|)^n}
\int_{\Lambda(a)}|f(z)|^q(1-|z|)^n\,d\omega(z).
\]
The tent space $T_q^\infty(\omega)$ consists of all measurable functions $f$ on $\mathbb B_n$ such that
\[
\|f\|_{T_q^\infty(\omega)}=\|C_{q,\,\omega}(f)\|_{L^\infty(\mathbb{S}_n)}.
\]
Let $\alpha>-n-1$. For $d\omega(z)=(1-|z|)^\alpha dV(z)$ we write $T_{q,\alpha}^p$ instead of  $T_q^p(\omega)$. 
The Hardy type tent space $HT_{q,\alpha}^p$ consists of $f \in H(\B)\cap T_{q,\alpha}^p$; 
and $HT_{\infty,\alpha}^p$ consists of $f \in H(\B)\cap T_{\infty,\alpha}^p$.

If $\omega=\sum_j\delta_{z_j}$ where $\{z_j\}$ is a separated sequence, and
$\delta_{z_j}$ are Dirac measures at $z_j$, then
we write $T_q^p(\{z_j\})$ instead of $T_q^p(\omega)$.
More precisely, for $0<p, q<\infty$ and for a separated sequence $Z=\{z_j\}_{j=1}^\infty$,
the tent sequence space $T_{q}^{p}(Z)$ consists of complex sequences 
$\lambda=\{\lambda_j\}_{k=1}^\infty$, where $\lambda_j=f(z_j)$, satisfying
$$
\|\lambda\|_{T_{q}^{p}(Z)}^{p}
:=\int_{\mathbb S_n}\left(\sum_{z_j\in\Gamma(\zeta)}
|\lambda_j|^{q}\right)^{\frac{p}{q}}\,d\sigma(\zeta)<\infty.
$$
Also, $\lambda=\{\lambda_j\} \in T_{\infty}^{p}(Z)$ if
$$
\|\lambda\|_{T_{\infty}^{p}(Z)}^{p}
:=\int_{\mathbb S_n} \left(\sup\{|\lambda_j|\,:\,z_j\in\Gamma(\zeta)\}\right)^p\,d\sigma(\zeta)<\infty.
$$
Finally, $\lambda=\{\lambda_j\}\in T_{q}^{\infty}(Z)$ if
$$
\|\lambda\|_{T_{q}^{\infty}(Z)}
=\sup_{\zeta\in\mathbb{S}_n}
\left(\sup_{a\in\Gamma(\zeta)}\frac{1}{(1-|a|)^n} 
\sum_{z_j\in\Lambda(a)}|\lambda_j|^{q}(1-|z_j|)^n\right)^{1/q}<\infty.
$$
It is well known that $\lambda\in T_{q}^{\infty}(Z)$ if and only if 
$\mu_{\lambda}=\sum_{j=1}^\infty|\lambda_j|^q(1-|z_j|)^n \delta_{z_j}$ 
is a Carleson measure.



We  will need the following discrete version's duality results for tent sequence spaces. 
See \cite{MA1999}, \cite{JM1996}, and \cite{HL1991} for the proofs.  

\begin{lem}\label{dual}
Let $1\le p<\infty$ and $Z=\{z_j\}$ be a separated sequence. 
If $1<q<\infty$, then the dual of $T_q^p(Z)$ is isomorphic to $T_{q'}^{p'}(Z)$ under the pairing
$$
\langle\lambda,\mu\rangle_{T_2^2(Z)}
=\sum_{j}\lambda_{j}\overline{\mu_{j}}(1-|z_j|^{2})^n, \ \  \mbox{where}\ 
\lambda\in T_{q}^{p}(Z), \ \mu\in T_{q'}^{p'}(Z).
$$
If $1<p<\infty$ and $0<q \leq 1$, then the dual of $T_{q}^{p}(Z)$ is isomorphic to $T_{\infty}^{p^{\prime}}(Z)$ under the same pairing.
\end{lem}

The following factorization of tent sequence spaces was proved by Miihkinen, Pau, Per\"al\" a and Wang in \cite{Mi}.

\begin{lem}\label{fact}
Assume  that $0<p, q<\infty$ and that  $Z=\{z_j\}$ is an $r$-lattice.  
Let $p\le p_{1}, p_{2}\le\infty, q\le q_{1}, q_{2}\le\infty$ satisfy
$$
\frac{1}{p}=\frac{1}{p_{1}}+\frac{1}{p_{2}}, 
\quad \frac{1}{q}=\frac{1}{q_{1}}+\frac{1}{q_{2}},
$$
and $p_iq_i>pq$, $i=1,2$.
Then, we have
$$
T_{q}^{p}(Z)=T_{q_{1}}^{p_{1}}(Z) \cdot T_{q_{2}}^{p_{2}}(Z);
$$
in other words, if 
$\gamma=\{\gamma_j\}\in T_{q_1}^{p_1}(Z)$ 
and 
$\beta=\{\beta_j\}\in T_{q_2}^{p_2}(Z)$, 
then   
$\gamma \cdot \beta=\{\gamma_j\beta_j\}_{j=1}^\infty\in T_{q}^{p}(Z)$ 
with 
$\|\gamma \cdot \beta\|_{T_{q}^{p}(Z)}\lesssim\|\gamma\|_{T_{q_1}^{p_1}(Z)}\cdot\|\beta\|_{T_{q_{2}}^{p_{2}}(Z)}$,
and conversely, if $\lambda \in T_{q}^{p}(Z)$, then there are sequences 
$\gamma \in T_{q_1}^{p_1}(Z)$ and $\beta \in T_{q_2}^{p_2}(Z)$ 
such that 
$\lambda=\gamma\cdot\beta$ and 
$\|\gamma\|_{T_{q_1}^{p_1}(Z)}\cdot\|\beta\|_{T_{q_2}^{p_2}(Z)}\lesssim\|\lambda\|_{T_{q}^{p}(Z)}$.
\end{lem}

The following result can be found in \cite{Mi}, Lemma 2 and Lemma 3.

\begin{lem}\label{3}
Let $0<p,q<\infty$. There exists $r_0\in(0,1)$ so that, if $0<r<r_0$ and $Z=\{z_j\}$ is an $r$-lattice, then
$$\int_{\mathbb{S}_n}\sup _{z\in\Gamma(\zeta)}|f(z)|^p(1-|z|)^\beta\,d\sigma(\zeta) 
\lesssim\int_{\mathbb{S}_n}\sup_{z_j\in\Gamma(\zeta)}|f(z_j)|^p(1-|z_j|)^{\beta}\,d\sigma(\zeta)$$ 
for $\beta\ge0$, and 
$$\int_{\mathbb{S}_n}\left(\int_{\Gamma(\zeta)}|f(z)|^p\,dV_\beta(z)\right)^q\,d\sigma(\zeta) 
\lesssim\int_{\mathbb{S}_n}\left(\sum_{z_j\in\Gamma(\zeta)}|f(z_j)|^p(1-|z_j|)^{n+1+\beta}\right)^q\,d\sigma(\zeta),$$
for $\beta>-n-1$, 
whenever $f\in H(\mathbb{B}_n)$ is such that the left-hand side is finite.
\end{lem}

\begin{lem}\label{l4.1}\cite[Proposition 7.2]{JP2016}
Let $Z=\{a_j\}$ be a separated sequence in $\mathbb{B}_n$ and let $0<p<\infty$.
If $b>n \max(2/p,1)$, then the operator $S_Z:\,T^p_2(Z)\rightarrow H^p$
defined by
\[
S_Z(\{\lambda_j\})(z)=\sum_j \lambda_j\frac{(1-|a_j|^2)^b}{(1-\langle z,a_j\rangle)^b}
\]
is bounded.
\end{lem}

The following result \cite{AB1988, FS1972, JP2016}
that gives an equivalent definition of Hardy spaces by tent spaces
is the version of the famous Calder\'on area theorem \cite{AC1965} for $\B$.

\begin{lem}\label{hp}
Let $0<p<\infty$ and $g\in H(\B)$. Then $g\in H^p$ if and only if
$\mathcal{R}g\in HT^p_{2,1-n}$. Moreover, if $g(0)=0$ then
$$
\|g\|_{H^p}^p\asymp \|\mathcal{R}g\|_{HT^p_{2,1-n}}^p
=\int_{\mathbb{S}_n}\left(\int_{\Gamma(\zeta)}|\mathcal{R}g(z)|^2(1-|z|)^2\,d\tau(z)\right)^{p/2}\,d\tau(\zeta).
$$
\end{lem}

\section{Boundedness}

In this section we prove Theorem~\ref{thm1.1}.
Due to the nature of the proofs, we break down this theorem into three parts.
First, we consider the case $0<p<q<\infty$, $s>0$ and the case $0<p=q<\infty$, $s\ge 2$. 

\begin{thm}\label{thm1.1a}
Let $k$ be a positive integer, let $0<p<q<\infty$, $s>0$; or $0<p=q<\infty$, $s\ge 2$. 
Fix an $r\in(0,1)$. 
Let $\mu$ be a positive Borel measure on $\mathbb{B}_n$ that is finite
on compact subsets of $\mathbb{B}_n$. 
Then 
$\A_{\,\mu,\,s}^k:H^p\rightarrow L^q(\mathbb{S}_n)$
is bounded if and only if $\mu$ is an $(ks/n+1+s/p-s/q)$-Carleson measure, that is,
\begin{equation}\label{t5.41}
\sup_{z\in\B}\frac{\mu(D(z,r))}{(1-|z|)^{ks+(1+s/p-s/q)n}}<\infty.
\end{equation}
\end{thm}

\begin{proof}
Suppose $\A_{\,\mu,\,s}^k:H^p\rightarrow L^q(\mathbb{S}_n)$ is bounded.
Fix $a\in\B$ and $0<r<1$. Let
$$
f_{a,p}(z)=\frac{(1-|a|^2)^{mn}}{(1-\langle a,z\rangle)^{(1/p+m)n}}.
$$
It can be easily checked that $f_{a,p}\in H^{p}$,
$\|f_{a,p}\|_{H^{p}}\asymp 1$ (uniformly for $a\in\mathbb{B}_n$).
Since $|1-\langle a,z\rangle|\asymp (1-|z|)\asymp(1-|a|^2)$ for $z\in D(a,r)$,
we get
$$
\left|\mathcal{R}^kf_{a,p}(z)\right|
\asymp\frac{(1-|a|^2)^{mn}}{|1-\langle a,z\rangle|^{(1/p+m)n+k}}
\asymp\frac1{(1-|a|^2)^{n/p+k}}.
$$
By the above estimate, (\ref{ei(z)}) and Lemma \ref{l51}, we get
\begr\label{t5.42}
&~&\left(\frac{\mu(D(a,r))}{(1-|a|)^{ks+(1+s/p-s/q)n}}\right)^{q/s}\\
&~&\qquad=(1-|a|)^n\left(\frac{\mu(D(a,r))}{(1-|a|)^{ks+(1+s/p)n}}\right)^{q/s}\nonumber\\
&~&\qquad\asymp\int_{I(a)}\left(\int_{D(a,r)}\left|\mathcal{R}^k f_{a,p}(z)\right|^s
\frac{d\mu(z)}{(1-|z|)^n}\right)^{q/s}d\sigma(\zeta)\nonumber\\
&~&\qquad\leq\int_{\mathbb{S}_n}\left(\int_{\widetilde{\Gamma}(\zeta)}\left|\mathcal{R}^kf_{a,p}(z)\right|^s
\frac{d\mu(z)}{(1-|z|)^n}\right)^{q/s}d\sigma(\zeta)\nonumber\\
&~&\qquad\asymp\|\A_{\,\mu,\,s}^kf_{a,p}\|_{L^q(\mathbb{S}_n)}^q
\lesssim \|\A_{\,\mu,\,s}^k\|_{H^p\rightarrow L^q(\mathbb{S}_n)}^q,\nonumber
\endr
where $I(a)=I_{\alpha}(a)=\{\zeta \in \mathbb S_n: a \in \Gamma_\alpha(\zeta)\},\ \alpha>2$.
Hence (\ref{t5.41}) is true.

Conversely, 
suppose that $\mu$ is an $(ks/n+1+s/p-s/q)$-Carleson measure.
we first consider the case $p<q$. 
Assume that (\ref{t5.41}) holds.
Let $r>0$ and $\{a_j\}_{j=1}^\infty$ be an $r$-lattice. 
By Lemma~\ref{dmk} with $m=0$ and Lemma~\ref{l51}, we have  
\begr\label{11}
(\A^k_{\mu, s}f(\zeta))^s
&\leq &\sum_{a_j \in \widetilde{\Gamma}(\zeta)}\int_{D(a_j, r)}|\mathcal{R}^kf(z)|^s\frac{d\mu(z)}{(1-|z|)^n}\\
&\lesssim&\sum_{a_j \in \widetilde{\Gamma}(\zeta)}\int_{D(a_j, r)}\frac{1}{(1-|z|)^{n+1+ks}}\int_{D(z, r)}| f(w)|^s\,dV(w)
\frac{d\mu(z)}{(1-|z|)^n}\nonumber\\
&\lesssim &\sum_{a_j \in \widetilde{\Gamma}(\zeta)}\int_{D(a_j, 2r)}|f(w)|^s\,dV(w)\frac{\mu(D(a_j, r))}{(1-|a_j|)^{2n+1+ks}}\nonumber\\
&\lesssim &\sum_{a_j \in \widetilde{\Gamma}(\zeta)}\int_{D(a_j, 2r)}|f(w)|^s(1-|w|)^{\alpha}\,dV(w)
\frac{\mu(D(a_j, r))}{(1-|a_j|)^{ks+n(1+s/p-s/q)}},\nonumber
\endr
where $\alpha=ns\left(\frac{1}{p}-\frac{1}{q}\right)-n-1$.
Clearly, $\alpha>-1-n$.
By Lemma \ref{l51} and (\ref{t5.41}) we get
\begr\label{11-0}
(\A^k_{\mu, s}f(\zeta))^s
&\lesssim&\int_{\widetilde{\widetilde{\Gamma}}(\zeta)}|f(w)|^s(1-|w|)^\alpha\,dV(w).
\endr
Choose $p<t<q$, then it is easy to see that 
$$
\left(\frac{t}{q}-1\right)n-1+\frac{t(n+1+\alpha)}{s}=\left(\frac{t}{p}-1\right)n-1.
$$
It is well-known that for $p<t$,  
\begin{equation}\label{h-b}
H^p \subseteq \mathcal{A}^t_{\left(\frac{t}{p}-1\right)n-1}.
\end{equation}
See, for example, Theorem 4.48 in \cite{KZ2005}.
By Lemma 3.1 in \cite{WZ2} and the well-known fact that 
$f\in \mathrm{A}^{t}_{\alpha}$ if and only if $\mathcal{R}f\in \mathrm{A}^t_{\alpha+t}$
(see, for example, Theorem 2.16 in \cite{KZ2005}),
we get that  
\begin{equation}\label{b-ht}
\mathrm{A}^t_{\left(\frac{t}{p}-1\right)n-1}
=\mathrm{A}^t_{\left(\frac{t}{q}-1\right)n-1+\frac{t\,(n+1+\alpha)}{s}} 
\subseteq HT^q_{s,\alpha}.
\end{equation}
Combining (\ref{h-b}) and (\ref{b-ht}) we get
\begin{equation}\label{h-ht}
H^p \subseteq HT^q_{s,\alpha},
\end{equation}
and a standard application of the Closed Graph Theorem shows that 
the inclusion is continuous. Thus, by \eqref{11-0} and \eqref{h-ht},
\begin{eqnarray*}
\|\A^k_{\mu, s}f\|_{L^q}^q
&=&\int_{\mathbb S_n}\left(\A^k_{\mu, s}f(\zeta))^s\right)^{q/s}d\sigma(\zeta)\\
&\lesssim&\int_{\mathbb S_n} \left(\int_{\widetilde{\widetilde{\Gamma}}(\zeta)}|f(w)|^s(1-|w|)^\alpha dV(w)\right)^{q/s}d\sigma(\zeta)\\
&\asymp&\|f\|^q_{HT^q_{s, \alpha}}\lesssim \|f\|^q_{H^p}.
\end{eqnarray*}
Hence, $\A_{\,\mu,\,s}^k:H^p\rightarrow L^q(\mathbb{S}_n)$ is bounded.

Next, we consider the case $p=q$, $s\geq 2$. By Lemma 2.2 and Fubini's theorem, we get
\begin{eqnarray}\label{Rin}
&~&(\A^k_{\mu, s}f(\zeta))^s\\
&~&\qquad\le\sum_{a_j \in \widetilde{\Gamma}(\zeta)}\int_{D(a_j,r)}|\mathcal{R}^kf(z)|^s\frac{d\mu(z)}{(1-|z|)^n}\nonumber\\
&~&\qquad\le\sum_{a_j \in \widetilde{\Gamma}(\zeta)}\int_{D(a_j r)}
|\mathcal{R}^kf(z)|^{s-2}|\mathcal{R}^kf(z)|^{2}\frac{d\mu(z)}{(1-|z|)^n}\nonumber\\
&~&\qquad\lesssim\sum_{a_j \in \widetilde{\Gamma}(\zeta)}\int_{D(a_j, r)}
\frac{1}{(1-|z|)^{k(s-2)}}\int_{D(z,r)}|f(w)|^{s-2}\,d\tau(w)\times\nonumber\\
&~&\qquad\quad\times\frac{1}{(1-|z|)^{2(k-1)}}\int_{D(z,r)}|\mathcal{R}f(w)|^2\,d\tau(w)
\frac{d\mu(z)}{(1-|z|)^n}\nonumber\\
&~&\qquad\lesssim(f^*(\zeta))^{s-2}\sum_{a_j \in \widetilde{\Gamma}(\zeta)}\int_{D(a_j,r)}
\int_{D(z,r)}|\mathcal{R}f(w)|^2(1-|w|)^{1-n}\,dV(w)\frac{d\mu(z)}{(1-|z|)^{n+ks}}\nonumber\\
&~&\qquad=(f^*(\zeta))^{s-2}\sum_{a_j\in\widetilde{\Gamma}(\zeta)}\int_{D(a_j,2r)}
|\mathcal{R}f(w)|^2(1-|w|)^{1-n}\int_{D(w,r)}\frac{d\mu(z)}{(1-|z|)^{n+ks}}\,dV(w)\nonumber\\
&~&\qquad\lesssim(f^*(\zeta))^{s-2}\sum_{a_j\in\widetilde{\Gamma}(\zeta)}\int_{D(a_j,2r)}
|\mathcal{R}f(w)|^2(1-|w|)^{1-n}\frac{\mu(D(w,r))}{(1-|w|)^{n+ks}}\,dV(w).\nonumber
\end{eqnarray}
By (\ref{t5.41}), Lemma~\ref{l51} and Lemma~\ref{covering}, we get
\begin{eqnarray*}
(\A^k_{\mu, s}f(\zeta))^s
&\lesssim&(f^*(\zeta))^{s-2}\sum_{a_j\in\widetilde{\Gamma}(\zeta)}\int_{D(a_j,2r)}|\mathcal{R}f(w)|^2(1-|w|)^{1-n}\,dV(w)\\
&\lesssim&(f^*(\zeta))^{s-2}\int_{\widetilde{\widetilde{\Gamma}}(\zeta)}|\mathcal{R}f(w)|^2(1-|w|)^{1-n}\,dV(w)
\end{eqnarray*}
By H\"older's inequality, Lemma~\ref{l2.41}, Lemma~\ref{hp} we get
\begin{eqnarray*}
\|\A^k_{\mu, s}f\|_{L^p(\mathbb{S}_n)}^p
&\lesssim&\int_{\mathbb{S}_n}(f^*(\zeta))^{p(s-2)/s}
\left(\int_{\widetilde{\widetilde{\Gamma}}(\zeta)}|\mathcal{R}f(w)|^2(1-|w|)^{1-n}\,dV(w)\right)^{p/s}\,d\sigma(\zeta)\\
&\lesssim&\left(\int_{\mathbb{S}_n}(f^*(\zeta))^{p}\,d\sigma(\zeta)\right)^{(s-2)/s}\times\\
&~&\qquad\times\left(\int_{\mathbb{S}_n}\left(\int_{\widetilde{\widetilde{\Gamma}}(\zeta)}|\mathcal{R}f(w)|^2(1-|w|)^{1-n}\,dV(w)\right)^{p/2}\,d\sigma(\zeta)\right)^{2/s}\\
&\lesssim&\|f\|_{H^p}^{p(s-2)/s}\cdot\|f\|_{H^p}^{2p/s}=\|f\|_{H^p}^p.
\end{eqnarray*}
Thus $\A_{\,\mu,\,s}^k:H^p\rightarrow L^p(\mathbb{S}_n)$ is bounded, and the proof is complete.
\end{proof}

For the other cases, we first prove a technical lemma.

\begin{lem}\label{test-b}
Let $k$ be a positive integer, let $0<p, q ,s<\infty$ and let $0<r<1$. 
Let $\mu$ be a positive Borel measure on $\mathbb{B}_n$ that is finite
on compact subsets of $\mathbb{B}_n$. 
Let $\Phi_{\mu}(z)=\mu(D(z,r))/(1-|z|)^{n+ks}$.
If $\mathcal{A}_{\mu, s}^k: H^p \to L^q(\mathbb{S}_n)$ is bounded then
for any $r$-lattice $Z=\{z_j\}\in\B$ and any complex sequence $\lambda=\{\lambda_j\}\in T^p_2(Z)$,
we have
\begin{equation}\label{lambda-phi}
\int_{\mathbb{S}_n}\left(\sum_{z_j\in \Gamma(\zeta)}|\lambda_j|^s\Phi_{\mu,2r}(z_j)\right)^{q/s}\,d\sigma(\zeta)
\lesssim\|\A_{\,\mu,\,s}^k\|_{H^p\rightarrow L^q(\mu)}^q\|\lambda\|_{T_2^p(Z)}^q.
\end{equation}
\end{lem}

\begin{proof}
Suppose $\A_{\,\mu,\,s}^k:H^p\rightarrow L^q(\mathbb{S}_n)$ is bounded. 
Let $Z=\{a_m\}$ be an $r$-lattice in $\B$. 
For  $t>\max(2/p,1)$, let $\lambda=\{\lambda_m\}\in T_2^p(Z)$ and 
$$
S_Z(\lambda)(z)=\sum_{m}\lambda_m\left(\frac{1-|a_m|}{1-\langle a_m,z\rangle}\right)^{nt},\ z\in\B.
$$
By Lemma~\ref{l4.1}, $S_{Z}:T_2^p(Z)\rightarrow H^p$ is bounded. 
Thus,
$$
\|\A_{\,\mu,\,s}^k(S_{Z}(\lambda))\|_{L^q(\mathbb{S}_n)}^q
\leq  \|\A_{\,\mu,\,s}^k\|_{H^p\rightarrow L^q(\mathbb{S}_n)}^q\|S_{Z}(\lambda)\|_{H^p}^q
\lesssim \|\A_{\,\mu,\,s}^k\|_{H^p\rightarrow L^q(\mathbb S_n)}^q\|\lambda\|_{T^p_2(Z)}^q.
$$
Hence,
$$
\int_{\mathbb{S}_n}\left(\int_{\Gamma(\zeta)}\left|\sum_{m}\lambda_m
\frac{(1-|a_m|)^{nt}}{(1-\langle a_m,z\rangle)^{nt+k}}\right|^s\frac{d\mu(z)}{(1-|z|)^n}\right)^{q/s}d\sigma(\zeta)
\lesssim \|\A_{\,\mu,\,s}^k\|_{H^p\rightarrow L^q(\mathbb S_n)}^q\|\lambda\|_{T_2^p(Z)}^q.
$$
Replacing $\lambda_m$ by $\lambda_mr_m(t)$, where $r_m$ denotes the $k$th Rademacher function,
and integrating with respect to $t$ to obtain
\begr
&~&\int_0^1\int_{\mathbb{S}_n}\left(\int_{\Gamma(\zeta)}\left|\sum_{m}\lambda_m
\frac{(1-|a_m|)^{nt}}{(1-\langle a_m,z\rangle)^{nt+k}}r_m(t)\right|^s\frac{d\mu(z)}{(1-|z|)^n}\right)^{q/s}\,d\sigma(\zeta)\,dt\nonumber\\
&~&\qquad\lesssim \|\A_{\,\mu,\,s}^k\|_{H^p\rightarrow L^q(\mathbb S_n)}^q\|\lambda\|_{T_2^p(Z)}^q.\nonumber
\endr
By an application of Fubini's theorem, Kahane's inequality and Khinchine's inequality, we have
\begr
&~&\int_{\mathbb{S}_n}\left(\int_{\Gamma(\zeta)}\left(\sum_{m}|\lambda_m|^2
\frac{(1-|a_m|)^{2nt}}{|1-\langle a_m,z\rangle|^{2nt+2k}}\right)^{s/2}\frac{d\mu(z)}{(1-|z|)^n}\right)^{q/s}d\sigma(\zeta)\nonumber\\
&~&\qquad\lesssim\|\A_{\,\mu,\,s}^k\|_{H^p\to L^q(\mathbb S_n)}^q\|\lambda\|_{T_2^p(Z)}^q.\nonumber
\endr
Thus, by Lemma~\ref{covering} (notice that each $z$ belongs to at most $N$ of the sets $D(a_j,2r)$), we have
\begr\label{abf}
&~&\|\A_{\,\mu,\,s}^k\|_{H^p\to L^q(\mu)}^q\|\lambda\|_{T_2^p(Z)}^q\nonumber\\
&~&\qquad\gtrsim\int_{\mathbb{S}_n}\left(\sum_{j}\int_{\Gamma(\zeta)\cap D(a_j,2r)}\left(\sum_{m}|\lambda_m|^2
\frac{(1-|a_m|)^{2nt}}{|1-\langle a_m,z\rangle|^{2nt+2k}}\right)^{s/2}\frac{d\mu(z)}{(1-|z|)^n}\right)^{q/s}d\sigma(\zeta)\nonumber\\
&~&\qquad\gtrsim\int_{\mathbb{S}_n}\left(\sum_{a_j\in \Gamma(\zeta)}\left(|\lambda_j|^2
\frac{(1-|a_j|)^{2nt}}{(1-|a_j|)^{2nt+2k}}\right)^{s/2}
\frac{\mu(D(a_j,2r))}{(1-|a_j|)^n}\right)^{q/s}d\sigma(\zeta)\nonumber\\
&~&\qquad=\int_{\mathbb{S}_n}\left(\sum_{a_j\in \Gamma(\zeta)}|\lambda_j|^s
\left(\frac{\mu(D(a_j,2r))}{(1-|a_j|)^{n+ks}}\right)\right)^{q/s}\,d\sigma(\zeta)\nonumber\\
&~&\qquad=\int_{\mathbb{S}_n}\left(\sum_{a_j\in \Gamma(\zeta)}|\lambda_j|^s\Phi_{\mu,2r}(a_j)\right)^{q/s}\,d\sigma(\zeta)\nonumber.
\endr
The proof is complete.
\end{proof}

Next, we consider the case $0<p=q<\infty$, $0<s<2$.

\begin{thm}\label{thm1.1b}
Let $k$ be a positive integer, let $0<p<\infty$, $0<s<2$, and let $0<r<1$. 
Let $\mu$ be a positive Borel measure on $\mathbb{B}_n$ that is finite
on compact subsets of $\mathbb{B}_n$. 
Let $\Phi_{\mu}(z)=\mu(D(z,r))/(1-|z|)^{n+ks}$.
Then the following statements are equivalent.
\begin{itemize}
\item[(i)] $\mathcal{A}^k_{\mu,s}: H^p \rightarrow L^p(\mathbb S_n)$ is bounded.
\item[(ii)] For some (or any) $r$-lattice $Z=\{a_j\}$ in $\B$, $\{\Phi_{\mu,2r}(a_j)\}\in T^{\infty}_{2/(2-s)}(Z)$.
\item[(iii)] $\Phi_{\mu}\in T^{\infty}_{2/(2-s)}(\tau)$.
\end{itemize}
\end{thm}

\begin{proof} We first prove that (i) and (ii) are equivalent.
Suppose (ii) holds for some $r$-lattice $Z=\{a_j\}$ in $\B$.
From Lemma~\ref{dmk} and H\"{o}lder's inequality, we get
\begin{eqnarray}\label{1}
(\A^k_{\mu, s}f(\zeta))^s
&\leq &\sum_{a_j\in\widetilde{\Gamma}(\zeta)}\int_{D(a_j,r)}|\mathcal{R}^kf(z)|^s\frac{d\mu(z)}{(1-|z|)^n}\nonumber\\
&\lesssim&\sum_{a_j\in\widetilde{\Gamma}(\zeta)}\int_{D(a_j,r)}
\frac{1}{(1-|z|)^{(k-1)s}}\left(\int_{D(z,r)}|\mathcal{R}f(w)|^2\,d\tau(w)\right)^{s/2}\frac{d\mu(z)}{(1-|z|)^n}\nonumber\\
&\lesssim&\sum_{a_j\in\widetilde{\Gamma}(\zeta)}\sup_{z\in D(a_j,r)}\left(\int_{D(z,r)}|\mathcal{R}f(w)|^2\,d\tau(w)\right)^{s/2}
\int_{D(a_j,r)}\frac{d\mu(z)}{(1-|z|)^{n+(k-1)s}}\nonumber\\
&\lesssim&\sum_{a_j\in\widetilde{\Gamma}(\zeta)}\left(\int_{D(a_j, 2r)}|\mathcal{R} f(w)|^2d\tau(w)\right)^{s/2}\frac{\mu(D(a_j, r))}{(1-|a_j|)^{n+(k-1)s}}\nonumber\\
&\asymp&\sum_{a_j \in \widetilde{\Gamma}(\zeta)}\left(\int_{D(a_j, 2r)}|\mathcal{R} f(w)|^2(1-|w|)^2d\tau(w)\right)^{s/2}
\frac{\mu(D(a_j, r))}{(1-|a_j|)^{n+ks}}.\nonumber
\end{eqnarray}
Let $\lambda=\{\lambda_j\}$ with 
$$
\lambda_j=\left(\int_{D(a_j, 2r)}|\mathcal{R} f(w)|^2(1-|w|)^2d\tau(w)\right)^{s/2},
$$
and let 
$$
\Phi_{\mu}(a_j)= \Phi_{\mu,r}(a_j)=\frac{\mu(D(a_j, r))}{(1-|a_j|)^{n+ks}}.
$$
Then, obviously, (ii) implies that $\{\Phi_{\mu}(a_j)\}\in T^{\infty}_{2/(2-s)}(Z)$.
From the above inequality we have
\begin{equation}\label{12}
(\A^k_{\mu, s}f(\zeta))^s
\lesssim\sum_{a_j \in \widetilde{\Gamma}(\zeta)}\lambda_j\cdot\,\Phi_{\mu}(a_j).
\end{equation}
By Lemma~\ref{hp}, we get that 
\begin{eqnarray*}
\|\lambda\|_{T^{p/s}_{2/s}(Z)}^{p/s}
&=&\int_{\mathbb{S}_n}\left(\sum_{a_j\in\Gamma(\zeta)}\int_{D(a_j,2r)}|\mathcal{R}f(w)|^2(1-|w|)^2\,d\tau(w)\right)^{p/2}\,d\sigma(\zeta)\\
&\lesssim& \int_{\mathbb{S}_n}\left(\int_{\widetilde{\Gamma}(\zeta)}|\mathcal{R}f(w)|^2(1-|w|)^2\,d\tau(w)\right)^{p/2}\,d\sigma(\zeta)\\
&\asymp&\|f\|_{H^p}^p.
\end{eqnarray*}
Thus $\lambda\in T^{p/s}_{2/s}(Z)$.
Since $\Phi_{\mu}(a_j)\in T^{\infty}_{2/(2-s)}(Z)$, by Lemma~\ref{fact}, we know that
$$
\left\{\lambda_j\cdot \Phi_{\mu}(a_j)\right\} \in T^{p/s}_{2/s}(Z)\cdot T^\infty_{{2}/{(2-s)}}(Z)=T^{p/s}_1(Z),\
$$
with
\begin{equation}\label{lambda-phi1}
\|\{\lambda_j\cdot \Phi_{\mu}(a_j)\}\|_{T^{p/s}_1(Z)}
\lesssim \|\{\Phi_{\mu}(a_j)\}\|_{T^{\infty}_{2/(2-s)}(Z)}\,\|\lambda\|_{T^{p/s}_{s/2}(Z)}
\lesssim \|\{\Phi_{\mu}(a_j)\}\|_{T^{\infty}_{2/(2-s)}(Z)}\,\|f\|_{H^p}^s.
\end{equation}
Putting these estimates together, we obtain that
\begin{eqnarray*}
\|\A^k_{\mu,s}(f)\|_{L^p(\mathbb{S}_n)}^p
&\lesssim&\int_{\mathbb{S}_n}\left(\sum_{a_j \in \widetilde{\Gamma}(\zeta)}\lambda_j\cdot \Phi_{\mu}(a_j)\right)^{p/s}d\sigma(\zeta)\\
&\asymp&\left\|\lambda_j\cdot \Phi_{\mu}(a_j) \right\|_{T^{p/s}_{1}(Z)}^{p/s}
\lesssim\|\{\Phi_{\mu}(a_j)\}\|_{T^{\infty}_{2/(2-s)}(Z)}^{p/s}\,\|f\|_{H^p}^p.
\end{eqnarray*} 
Thus, $\A^k_{\mu,s}: H^p\to L^p(\mathbb S_n)$ is bounded, or (i) is true.

Conversely, suppose that (i) holds, that is, $\A^k_{\mu,s}: H^p\to L^p(\mathbb S_n)$ is bounded.
Let $Z=\{a_j\}$ be any $r$-lattice in $\B$. 
By Lemma~\ref{test-b},
for any complex sequence $\lambda=\{\lambda_j\}\in T^p_2(Z)$ we have
$$
\int_{\mathbb{S}_n}\left(\sum_{z_j\in \Gamma(\zeta)}|\lambda_j|^s\Phi_{\mu,2r}(z_j)\right)^{q/s}\,d\sigma(\zeta)
\lesssim\|\A_{\,\mu,\,s}^k\|_{H^p\to L^q(\mu)}\|\lambda\|_{T_2^p(Z)}^q.
$$
Let $\rho>1$ be large enough such that $p\rho>s$.
By Lemma~\ref{dual} and Lemma~\ref{fact} we have 
$$
T^{\infty}_{{2\rho}/{(2-s)}}(Z)
=\left(T^{1}_{({2\rho}/{(2-s)})'}(Z)\right)^\ast
=\left(T^{p\rho / (p\rho-s)}_{\rho'}(Z)\cdot T^{p\rho/s}_{2\rho/s}(Z)\right)^\ast.
$$
Thus, given any $\varphi=\{\varphi_j\} \in T^{1}_{({2\rho}/{(2-s)})'}(Z)$, 
we can write $\varphi_j=\gamma_j\lambda_j^{s/\rho}$, where
$\gamma=\{\gamma_j\}\in T^{{p\rho}/{(p\rho-s)}}_{\rho'}(Z)$
and 
$\lambda=\{\lambda_j\}\in T^p_2(Z)$.
Let $\varphi$ be a measurable function on $\B$ such that $\varphi(z_j)=\varphi_j$,
and $d\mu(z)=\sum_jd\delta_{z_j}(z)$, where $\delta_{z_j}$ is the Dirac measure at $z_j$.
By Fubini's theorem, 
\begin{eqnarray*}
\sum_j|\varphi_j|\,\Phi^{1/\rho}_{\mu,2r}(z_j)(1-|z_j|)^n
&\asymp&\int_{\B}|\varphi(z)|\Phi^{1/\rho}_{\mu,2r}(z)\int_{I(z)}\,d\sigma(\zeta)\,d\mu(z)\\
&\asymp&\int_{\mathbb{S}_n}\int_{\Gamma(\zeta)}|\varphi(z)|\Phi^{1/\rho}_{\mu,2r}(z)\,d\mu(z)\,d\sigma(\zeta)\\
&\asymp&\int_{\mathbb{S}_n}\sum_{z_j\in\Gamma(\zeta)}|\varphi_j|\,\Phi^{1/\rho}_{\mu,2r}(z_j)\,d\sigma(\zeta)\\
&\asymp&\int_{\mathbb{S}_n}\sum_{z_j\in\Gamma(\zeta)}|\gamma_j\lambda_j^{s/\rho}|\,\Phi^{1/\rho}_{\mu,2r}(z_j)\,d\sigma(\zeta).
\end{eqnarray*}
By H\"{o}lder's inequality and Lemma~\ref{test-b}, we get
\begin{eqnarray}\label{4}
&~&\sum_j|\varphi_j|\,\Phi^{1/\rho}_{\mu,2r}(z_j)(1-|z_j|)^n\\
&~&\qquad\leq \int_{\mathbb{S}_n}\left(\sum_{z_j\in \Gamma(\zeta)} |\gamma_j|^{\rho'}\right)^{1/\rho'}
\left(\sum_{z_j\in \Gamma(\zeta)}|\lambda_j|^s \Phi_{\mu,2r}(z_j)\right)^{1/\rho}\,d\sigma(\zeta)\nonumber\\
&~&\qquad\leq\|\gamma\|_{T_{\rho'}^{p\rho/(p\rho-s)}(Z)}\left(\int_{\mathbb{S}_n}\left(\sum_{z_j\in \Gamma(\zeta)}
|\lambda_j|^s \Phi_{\mu,2r}(z_j)\right)^{q/s}\,d\sigma(\zeta)\right)^{s/q\rho} \nonumber\\
&~&\qquad\lesssim\|\gamma\|_{T_{\rho'}^{p\rho/(p\rho-s)}(Z)}
\|\mathcal{A}^k_{\mu,s}\|^{s/\rho}_{H^p\rightarrow L^q\mathbb{S}_n)}\|\lambda\|_{T^p_2(Z)}^{s/\rho}\nonumber\\
&~&\qquad\lesssim\|\mathcal{A}^k_{\mu,s}\|^{s/\rho}_{H^p\rightarrow L^q(\mathbb{S}_n)}\|\varphi\|_{T^1_{({2\rho}/{(2-s)})'}(Z)}.\nonumber
\end{eqnarray}
By Lemma~\ref{dual} we obtain that 
$\Phi_{\mu,2r}^{1/{\rho}}(z_j)\in T^{\infty}_{2\rho/(2-s)}(Z)$,
or equivalently,
$\Phi_{\mu,2r}(z_j)\in T^{\infty}_{2/(2-s)}(Z)$. Thus (ii) is true.

Next, we prove that (ii) implies (iii). 
Fix $0<r<1$.
Suppose (ii) holds for some $r$-lattice $\{a_j\}$ in $\B$.
It is equivalent to that
$d\omega(z)=\sum_{j=1}^{\infty}|\Phi_{\mu,2r}(a_j)|^{2/(2-s)}(1-|a_j|)^n\,d\delta_{a_j}(z)$
is a Carleson measure.
Let $d\widetilde{\omega}(z)=|\Phi_{\mu}(z)|^{2/(2-s)}(1-|z|)^n\,d\tau(z)$.
Then, by Lemma~\ref{covering}, Lemma 2.20 and Lemma 2.27 in \cite{KZ2005},
we obtain that, for any $a\in\B$,
\begin{eqnarray}\label{carleson-cd}
&~&\int_{\B}\frac{(1-|a|)^n}{|1-\langle a,z\rangle|^{2n}}\,d\widetilde{\omega}(z)\\
&~&\qquad\lesssim\sum_{j=1}^{\infty}\int_{D(a_j,r)}\frac{(1-|a|)^n}{|1-\langle a,z\rangle|^{2n}}
\left(\frac{\mu(D(z,r))}{(1-|z|)^{n+ks}}\right)^{2/(2-s)}(1-|z|)^n\,d\tau(z)\nonumber\\
&~&\qquad\lesssim\sum_{j=1}^{\infty}\frac{(1-|a|)^n}{|1-\langle a,a_j\rangle|^{2n}}
\left(\frac{\mu(D(a_j,2r))}{(1-|a_j|)^{n+ks}}\right)^{2/(2-s)}(1-|a_j|)^n\nonumber\\
&~&\qquad\asymp\int_{\B}\frac{(1-|a|)^n}{|1-\langle a,z\rangle|^{2n}}\,d\omega(z).\nonumber
\end{eqnarray}
By Lemma~\ref{carleson},
we get that $d\widetilde{\omega}(z)$ is a Carleson measure,
which is equivalent to that
$$
\|\Phi_{\mu}\|_{T_{2/(2-s)}^{\infty}(\tau)}
=\sup_{a\in\B}\frac{1}{(1-|a|)^n} \int_{\Lambda(a)}|\Phi_{\mu}(z)|^{2/(2-s)}(1-|z|)^n\,d\tau(z)<\infty.
$$
Thus (iii) is true.

Finally, we prove that (iii) implies (i) with $p=s<2$,
that is, $\A_{\mu,s}^k:\,H^s\to L^s(\mathbb{S}_n)$ is bounded.
Since we have proved that (i) (with $p=s<2$) implies (ii),
this will show that (iii) and (ii) are equivalent, and thus complete the proof of the theorem.
Suppose that (iii) holds. 
For any $f\in H^s$, let $L(w)=|\mathcal{R}f(w)|^s(1-|w|)^s$. 
By H\"older's inequality and Lemma~\ref{hp}, we get
\begin{eqnarray*}
\|L\|_{T^1_{2/s}(\tau)}
=\int_{\mathbb{S}_n}\left(\int_{\Gamma(\zeta)}|\mathcal{R}f(w)|^{2}(1-|w|)^2\,d\tau(w)\right)^{s/2}\,d\sigma(\zeta)
\lesssim\|f\|_{H^s}^s.
\end{eqnarray*}
Thus, by Lemma~\ref{dmk}, Fubini's theorem and the duality of tent spaces 
(see, for example, Lemma 2.3 in \cite{LLZ2022}) we get,
\begin{eqnarray*}
&~&\|\A_{\mu,s}^k(f)\|_{L^s(\mathbb{S}_n)}^s\\
&~&\qquad=\int_{\mathbb{S}_n}\int_{\Gamma(\zeta)}|\mathcal{R}^k f(z)|^s\frac{d\mu(z)}{(1-|z|)^n}\,d\sigma(\zeta)\\
&~&\qquad\lesssim\int_{\mathbb{S}_n}\int_{\Gamma(\zeta)}
\frac{1}{(1-|z|)^{(k-1)s}}\int_{D(z,r)}|\mathcal{R} f(w)|^s\,d\tau(w)
\frac{d\mu(z)}{(1-|z|)^n}\,d\sigma(\zeta)\\
&~&\qquad\lesssim\int_{\mathbb{S}_n}\int_{\widetilde{\Gamma}(\zeta)}|\mathcal{R}f(w)|^s(1-|w|)^s
\Phi_{\mu}(w)\,d\tau(w)\,d\sigma(\zeta)\\
&~&\qquad\asymp\int_{\mathbb{B}_n}|\mathcal{R}f(w)|^s(1-|w|)^s\Phi_{\mu}(w)\int_{\widetilde{\Lambda}(w)}\,d\sigma(\zeta)\,d\tau(w)\\
&~&\qquad\asymp\int_{\mathbb{B}_n} H(w)\Phi_{\mu}(w)(1-|w|)^n\,d\tau(w)\\
&~&\qquad\lesssim\|H\|_{T^1_{2/s}(\tau)}\|\Phi_{\mu}\|_{T^{\infty}_{2/(2-s)}(\tau)}
\lesssim\|f\|_{H^s}^s\|\Phi_{\mu}\|_{T^{\infty}_{2/(2-s)}(\tau)}.
\end{eqnarray*}
Here $\widetilde{\Lambda}(w)=\{\zeta\in\mathbb{S}_n:\,w\in\widetilde{\Gamma}(\zeta)\}$.
Thus $\mathcal{A}^k_{\mu,s}: H^s \rightarrow L^s(\mathbb S_n)$ is bounded, or (i) holds for $p=s<2$.
The proof is complete.
\end{proof}

\begin{rem}\label{r-thm1.1b}
Carefully examining the proof, we can see that 
$\|\Phi_{\mu}\|_{T^{\infty}_{2/(2-s)}(\tau)}
\asymp\|\Phi_{\mu,2r}\|_{T^{\infty}_{2/(2-s)}(Z)}$.
\end{rem}

Our last result in this section deals with the case $0<q<p<\infty$.

\begin{thm}\label{thm1.1c}
Let $k$ be a positive integer, let $0<p, q ,s<\infty$ and let $0<r<1$. 
Let $\mu$ be a positive Borel measure on $\mathbb{B}_n$ that is finite
on compact subsets of $\mathbb{B}_n$. 
Let $\Phi_{\mu}(z)=\mu(D(z,r))/(1-|z|)^{n+ks}$.
Then the following statements hold.
\begin{itemize}
\item[(i)] If  $p>q$ and $s \geq 2$, then $\mathcal{A}_{\mu, s}^k: H^p \to L^q(\mathbb{S}_n)$
 is bounded if and only if $\Phi_{\mu} \in T_{\infty}^{pq/(s(p-q))}(\tau)$.
\item[(ii)] If $p>q$ and  $s <2$, then $\mathcal{A}_{\mu,s}^k: H^p \to L^q(\mathbb{S}_n)$ is bounded 
if and only if $\Phi_{\mu}\in T_{2/(2-s)}^{pq/(s(p-q))}(\tau)$.
\end{itemize}
\end{thm}

\begin{proof}
(i). Suppose $p>q$ and $s\ge 2$. Let $\Phi_{\mu}\in T_\infty^{pq/((p-q)s)}(\tau)$,
and let $0<r<1$. 
By Lemma~\ref{dmk}, for $\zeta\in\partial\mathbb{S}_n$ and $w\in\Gamma(\zeta)$,
$$
|\mathcal{R} f(w)|(1-|w|)
\lesssim\int_{D(w,r)}|f(z)|\,d\tau(z)
\lesssim f^*(\zeta).
$$
By Lemma~\ref{dmk}, Lemma~\ref{l51} and and Fubini's theorem we obtain
\begin{eqnarray}\label{fwr}
&~&\|\A_{\mu,s}^k(f)\|_{L^q(\mathbb{S}_n)}^q\\
&~&\qquad=\int_{\mathbb{S}_n}\left(\int_{\Gamma(\zeta)}|\mathcal{R}^k f(z)|^s
\frac{d\mu(z)}{(1-|z|)^n}\right)^{q/s}\,d\sigma(\zeta)\nonumber\\
&~&\qquad\lesssim\int_{\mathbb{S}_n}\left(\int_{\Gamma(\zeta)}
\frac{1}{(1-|z|)^{(k-1)s}}\int_{D(z,r)}|\mathcal{R} f(w)|^s\,d\tau(w)
\frac{d\mu(z)}{(1-|z|)^n}\right)^{q/s}\,d\sigma(\zeta)\nonumber\\
&~&\qquad\lesssim\int_{\mathbb{S}_n}\left(\int_{\widetilde{\Gamma}(\zeta)}|\mathcal{R}f(w)|^s(1-|w|)^s
\Phi_{\mu}(w)\,d\tau(w)\right)^{q/s}\,d\sigma(\zeta).\nonumber
\end{eqnarray}
By H\"older's inequality we get
\begin{eqnarray*}
&~&\|\A_{\mu,s}^k(f)\|_{L^q(\mathbb{S}_n)}^q\\
&~&\qquad\lesssim\int_{\mathbb{S}_n}|f^*(\zeta)|^{(s-2)q/s}\left(\int_{\widetilde{\Gamma}(\zeta)}|\mathcal{R} f(w)|^2(1-|w|)^2
\Phi_{\mu}(w)\,d\tau(w)\right)^{q/s}\,d\sigma(\zeta)\nonumber\\
&~&\qquad\lesssim\left(\int_{\mathbb{S}_n}|f^*(\zeta)|^p\,d\sigma(\zeta)\right)^{(s-2)q/(sp)}\times\nonumber\\
&~&\qquad\quad\times\left(\int_{\mathbb{S}_n}\left(\int_{\widetilde{\Gamma}(\zeta)}|\mathcal{R} f(w)|^2(1-|w|)^2
\Phi_{\mu}(w)\,d\tau(w)\right)^{ {pq}/{(sp-(s-2)q)}}\,d\sigma(\zeta)\right)^{(sp-(s-2)q)/(sp)}.\nonumber
\end{eqnarray*}
By H\"{o}lder's inequality again we get
\begin{eqnarray*}\label{fwr1}
&~&\int_{\mathbb{S}_n}\left(\int_{\widetilde{\Gamma}(\zeta)}|\mathcal{R} f(w)|^2(1-|w|)^2
\Phi_{\mu}(w)\,d\tau(w)\right)^{pq/(sp-(s-2)q)}\,d\sigma(\zeta)\\
&~&\qquad\leq\int_{\mathbb{S}_n}\left(\sup_{w\in \widetilde{\Gamma}(\zeta)}\Phi_{\mu}(w)\right)^{pq/(sp-(s-2)q)}\times\nonumber\\
&~&\qquad\quad\times\left(\int_{\widetilde{\Gamma}(\zeta)}|\mathcal{R} f(w)|^2(1-|w|)^2\,d\tau(w)\right)^{pq/(sp-(s-2)q)}\,d\sigma(\zeta)\nonumber\\
&~&\qquad\leq\left(\int_{\mathbb{S}_n}\left(\sup_{w\in\widetilde{\Gamma}(\zeta)}\Phi_{\mu}(w)\right)^{pq/(s(p-q))}\,
d\sigma(\zeta)\right)^{s(p-q)/(sp-(s-2)q)}\times\nonumber\\
&~&\qquad\quad\times\left(\int_{\mathbb{S}_n}\left(\int_{\widetilde{\Gamma}(\zeta)}
|\mathcal{R}f(w)|^2(1-|w|)|^2\,d\tau(w)\right)^{p/2}\,d\sigma(\zeta)\right)^{2q/(sp-(s-2)q)}.\nonumber
\end{eqnarray*}
Combining the above two inequalities, by Lemma~\ref{l2.41} and Lemma~\ref{hp} we obtain that
$$
\|\A_{\,\mu,\,s}^k(f)\|_{L^q(\mathbb{S}_n)}^q
\lesssim \|f\|_{H^p}^q\,\|\Phi_{\mu}\|_{T_\infty^{pq/(s(p-q))}(\tau)}^{q/s}.
$$
Hence, $\A_{\,\mu,\,s}^k:H^p\rightarrow L^q(\mathbb{S}_n)$ is bounded.

Conversely, 
by Lemma~\ref{3},
we only need to prove that 
$\{\Phi_{\mu,2r}(z_j)\}\in T_{\infty}^{pq/(s(p-q))}(Z)$ for any $r$-lattice $Z=\{z_j\}$. 
Choose $\rho>1$ such that $q\rho>s$. Then $pq\rho/(s(p-q)>1$.
We just need to prove that $\{\Phi^{1/\rho}_{\mu,2r}(z_j)\}\in T_{\infty}^{pq\rho/(s(p-q))}(Z)$.
Let $\sigma>0$ satisfy
$$
\frac{1}{\sigma}=\frac{1}{\rho'}+\frac{s}{2\rho}.
$$
Then $\sigma \leq 1$ since $s\geq 2.$ 
By Lemma~\ref{dual} and Lemma~\ref{fact} we get 
$$
T^{{pq\rho/(s(p-q))}}_{\infty}(Z)
=\left(T^{({pq\rho/(s(p-q))})'}_{\sigma}(Z)\right)^\ast
=\left(T^{{q\rho/(q\rho-s)}}_{\rho'}(Z)\cdot T^{p\rho/s}_{2\rho/s}(Z)\right)^\ast.
$$
Using Lemma~\ref{test-b},
similar to the proof of \eqref{4}, we can get that 
$$
\sum_j|\varphi_j|\,\Phi^{1/\rho}_{\mu,2r}(z_j)(1-|z_j|)^n
\lesssim\|\mathcal{A}^k_{\mu,s}\|^{s/\rho}_{H^p\to L^q(\mathbb{S}_n)}\|\varphi\|_{T^{({pq\rho/(s(p-q))})'}_{\sigma}(Z)}
$$
for any $\varphi=\{\varphi_j\}\in T^{({pq\rho/(s(p-q))})'}_{\sigma}(Z)$.
By Lemma~\ref{dual} we obtain that 
$\{\Phi_{\mu,2r}(z_j)\}\in T_{\infty}^{pq/(s(p-q))}(Z)$.

(ii). 
Suppose $p>q$ and $0<s<2$.
Suppose $\Phi_{\mu}\in T_{2/(2-s)}^{pq/(s(p-q))}(\tau)$. 
By (\ref{fwr}) and H\"{o}lder's inequality we have
\begin{eqnarray*}
&~&\|\A_{\mu,s}^k(f)\|_{L^q(\mathbb{S}_n)}^q\\
&~&\qquad\lesssim\int_{\mathbb{S}_n}\left(\int_{\widetilde{\Gamma}(\zeta)}|\mathcal{R}f(w)|^s(1-|w|)^s
\Phi_{\mu}(w)\,d\tau(w)\right)^{q/s}d\sigma(\zeta)\\
&~&\qquad\lesssim\int_{\mathbb{S}_n}\left(\int_{\widetilde{\Gamma}(\zeta)}|\mathcal{R}f(w)|^2(1-|w|)^2\,d\tau(w)\right)^{q/2}\times\\
&~&\qquad\quad\times\left(\int_{\widetilde{\Gamma}(\zeta)}\left(\Phi_{\mu}(w)\right)^{2/(2-s)}d\tau(w)\right)^{(2-s)q/2s}d\sigma(\zeta)\\
&~&\qquad\lesssim\left(\int_{\mathbb{S}_n}\left(\int_{\widetilde{\Gamma}(\zeta)}|\mathcal{R}f(w)|^2(1-|w|)^2\,
d\tau(w)\right)^{p/2}d\sigma(\zeta)\right)^{q/p}\times\\
&~&\qquad\quad\times\left(\int_{\mathbb{S}_n}
\left(\int_{\widetilde{\Gamma}(\zeta)}\left(\Phi_{\mu}(w)\right)^{2/(2-s)}d\tau(w)\right)^{((2-s)/2)pq/(s(p-q))}d\sigma(\zeta)\right)^{(p-q)/p}.
\end{eqnarray*}
Therefore, by Lemma~\ref{hp} we get
\[
\|\A_{\,\mu,\,s}^k(f)\|_{L^q(\mathbb{S}_n)}\lesssim \|f\|_{H^p}\|\Phi_{\mu,r}\|_{T_{2/(2-s)}^{pq/(s(p-q))}}^{1/s}.
\]
Thus, $\A_{\,\mu,\,s}^k:H^p\rightarrow L^q(\mathbb{S}_n)$ is bounded.

Conversely,  
again, 
By Lemma \ref{3}, 
we only need to prove that 
$\{\Phi_{\mu,2r}(z_j)\}\in T_{2 /(2-s)}^{pq/(s(p-q))}(Z)$ for any $r$-lattice $Z=\{z_j\}$ in $\B$, 
which is equivalent to
$$
\{\Phi^{1/\rho}_{\mu,2r}(z_j)\}
\in T_{2\rho /(2-s)}^{pq\rho/(s(p-q))}(Z),
$$ 
where $\rho>1$ such that $q\rho>s$.
By Lemma~\ref{dual} and Lemma~\ref{fact} we get that
$$
T^{{pq\rho/(s(p-q))}}_{{2\rho/(2-s)}}(Z)
=\left(T^{({pq\rho/(s(p-q))})'}_{({2\rho/(2-s)})'}(Z)\right)^*
=\left(T^{{q\rho/(q\rho-s)}}_{\rho'}(Z)\cdot T^{p\rho/s}_{2\rho/s}(Z)\right)^*.
$$
Using Lemma~\ref{test-b},
similar to the proof of \eqref{4}, we can get that 
$$
\sum_j|\varphi_j|\,\Phi^{1/\rho}_{\mu,2r}(z_j)(1-|z_j|)^n
\lesssim\|\mathcal{A}^k_{\mu,s}\|^{s/\rho}_{H^p\rightarrow L^q(\mathbb{S}_n)}
\|\varphi\|_{T^{({pq\rho/(s(p-q))})'}_{({2\rho/(2-s)})'}(Z)}
$$
for any $\varphi=\{\varphi_j\}\in T^{({pq\rho/(s(p-q))})'}_{({2\rho/(2-s)})'}(Z)$.
By Lemma~\ref{dual} we obtain that 
$\{\Phi_{\mu,2r}(z_j)\}\in T^{{pq\rho/(s(p-q))}}_{{2\rho/(2-s)}}(Z)$.
The proof is complete.
\end{proof}

The proofs of Theorem~\ref{thm1.1a}, Theorem~\ref{thm1.1b} and  Theorem~\ref{thm1.1c} actually provide
alternative proofs of the following well-known result on embedding derivatives of Hardy spaces into Lebesgue spaces. 
The result for part (i) below was first proved by Luecking in \cite{HL1985} for the unit disk,
and then by Kang and Koo in \cite{KK2000} for the unit ball (in a more general form).
The result for other cases was first proved by  Luecking in \cite{HL1991} for $\mathbb{R}^n$.
See also \cite{JM1995}, \cite{JM1996} and \cite{MA1999} for some corresponding results
on embedding derivatives of  $\mathcal{M}$-harmonic Hardy spaces into Lebesgue spaces.

\begin{cor}\label{b-cm}
Let $0<p, q ,s<\infty$, let $k \in \mathbb N$, let $0<r<1$ and 
let $\mu$ be a positive Borel measure on $\mathbb{B}_n$, finite on compact subsets of $\B$. 
Let $\Phi_{\mu}(z)=\mu(D(z,r))/(1-|z|)^{n+kq}$.
Then the following statements hold.
\begin{itemize}
\item[(i)] If $p<q$ or $p=q\ge 2$ 
then $\mathcal{R}^k: H^p\to L^q(\mu)$  is bounded if and only if 
$$
\sup_{z\in\B}\frac{\mu(D(z,r))}{(1-|z|)^{kq+nq/p}}<\infty.
$$
\item[(ii)] If $p=q<2$
then $\mathcal{R}^k: H^p\to L^q(\mu)$  is bounded if and only if 
$
\Phi_{\mu}\in T^{\infty}_{2/(2-q)}(\tau).
$
\item[(iii)] If $p>q\ge 2$ 
then $\mathcal{R}^k: H^p\to L^q(\mu)$  is bounded if and only if 
$
\Phi_{\mu}\in T^{p/(p-q)}_{\infty}(\tau).
$
\item[(iv)] If $p>q$ and $q<2$ 
then $\mathcal{R}^k: H^p\to L^q(\mu)$  is bounded if and only if 
$
\Phi_{\mu}\in T^{p/(p-q)}_{2/(2-q)}(\tau).
$
\end{itemize}
\end{cor}

\begin{proof}
By Fubini's theorem we have
\begin{eqnarray}\label{Aps}
\|\A_{\mu,\,q}^k(f)\|_{L^q(\mathbb{S}_n)}^q
&=&\int_{\mathbb{S}_n}\int_{\Gamma(\zeta)}|\mathcal{R}^kf(z)|^q\frac{d\mu(z)}{(1-|z|)^n}\,d\sigma(\zeta)\\
&=&\int_{\mathbb{B}_n}|\mathcal{R}^kf(z)|^q\int_{\mathbb{S}_n}\chi_{I(z)}(\zeta)\,d\sigma(\zeta)\frac{d\mu(z)}{(1-|z|)^n}\nonumber\\
&\asymp&\int_{\mathbb{B}_n}|\mathcal{R}^kf(z)|^q\,d\mu(z)\nonumber.
\end{eqnarray}
Thus 
$\mathcal{R}^k: H^p\rightarrow L^q(\mu)$ is bounded
if and only if 
$\A_{\mu,\,q}^k:H^p\rightarrow L^q(\mathbb{S}_n)$ is bounded. 
The result is then obtained by Theorem~\ref{thm1.1a}, Theorem~\ref{thm1.1b} and  Theorem~\ref{thm1.1c}  with $s=q$.
\end{proof}


\section{Compactness}

We need the following result. Similar results for many bounded linear operators
are easy to prove. However, for the sublinear operator $\A_{\mu,s}^k$, the proof
is surprisingly involved.

\begin{lem}\label{compact}
Let $0<p\le q<\infty$, let $0<s<\infty$, and $\mu$ be a positive Borel measure on $\B$.
Suppose that $\A_{\mu,s}^k:H^p\rightarrow L^q(\mathbb{S}_n)$ is bounded.
Then the following statements are equivalent.
\begin{itemize}
\item[(i)] $\A_{\mu,s}^k:H^p\rightarrow L^q(\mathbb{S}_n)$ is compact.
\item[(ii)] For any bounded sequence $\{f_i\}$ in $H^p$ such that $\{f_i\}$
converges to 0 uniformly on compact subsets of $\mathbb{B}_n$, we have
$$
\lim_{i\to\infty}\|\A_{\mu,s}^k(f_i)\|_{L^q(\mathbb{S}_n)}=0.
$$
\end{itemize}
\end{lem}

\begin{proof} (i)$\Rightarrow$(ii). 
Let $0<p\le q<\infty$.
Assume that $\A_{\mu,s}^k:H^p\rightarrow L^q(\mathbb{S}_n)$ is compact.
Since $L^q(\mathbb{S}_n)\subseteq L^p(\mathbb{S}_n)$, and this embedding is continuous, 
we know that $\A_{\mu,s}^k:H^p\rightarrow L^p(\mathbb{S}_n)$ is compact.
Let $\{f_i\}$ be a bounded sequence in $H^p$ such that $\{f_i\}$
converges to 0 uniformly on compact subsets of $\mathbb{B}_n$.
Choose a real number $t$ such that $0<t<s$. Similarly to the proof of (\ref{Rin})
we get
\begin{equation}\label{c1}
(\A^k_{\mu, s}f_i(\zeta))^s
\lesssim(f_i^*(\zeta))^{s-t}
\int_{\Gamma(\zeta)}|\mathcal{R}f_i(w)|^t(1-|w|)^{t}\frac{\mu(D(w,r))}{(1-|w|)^{n+ks}}\,d\tau(w).
\end{equation}
Since $\A_{\mu,s}^k:H^p\rightarrow L^p(\mathbb{S}_n)$ is compact, it is also bounded.
By Theorem~\ref{thm1.1a}, we know that
\begin{equation}\label{c2}
\sup_{w\in\B}\frac{\mu(D(w,r))}{(1-|w|)^{n+ks}}<\infty.
\end{equation}
Choose a real number $\eta$ small enough such that $0<\eta<t$. Then
\begin{equation*}
\lim_{|w|\to1}\frac{\mu(D(w,r))}{(1-|w|)^{n+ks-\eta}}=0.
\end{equation*}
Thus, there is a number $\delta$, $0<\delta<1$, such that
\begin{equation}\label{c3}
\sup_{w\in\B\setminus\overline{B(0,\delta)}}\frac{\mu(D(w,r))}{(1-|w|)^{n+ks}}<\e^s.
\end{equation}
Let
$$
I_1=\int_{\Gamma(\zeta)\cap\overline{B(0,\delta)}}|\mathcal{R}f_i(w)|^t(1-|w|)^{t}\frac{\mu(D(w,r))}{(1-|w|)^{n+ks}}\,d\tau(w)
$$
and
\begin{eqnarray*}
I_2&=&\int_{\Gamma(\zeta)\setminus\overline{B(0,\delta)}}|\mathcal{R}f_i(w)|^t(1-|w|)^{t}
\frac{\mu(D(w,r))}{(1-|w|)^{n+ks}}\,d\tau(w)\\
&=&\int_{\Gamma(\zeta)\setminus\overline{B(0,\delta)}}|\mathcal{R}
f_i(w)|^t(1-|w|)^{t-\eta}\frac{\mu(D(w,r))}{(1-|w|)^{n+ks-\eta}}\,d\tau(w)
\end{eqnarray*}
Then (\ref{c1}) can be written as
\begin{equation}\label{c4}
(\A^k_{\mu, s}f_i(\zeta))^s
\lesssim(f_i^*(\zeta))^{s-t}(I_1+I_2).
\end{equation}
By (\ref{c3}), Lemma~\ref{l2.41} and H\"older's inequality we get
\begin{eqnarray}\label{c5}
&~&\int_{\mathbb{S}_n}[(f_i^*(\zeta))^{(s-t)/s}I_2^{1/s}]^{p/2}\,d\sigma(\zeta)\\
&~&\qquad\lesssim\e^{p/2}\int_{\mathbb{S}_n}(f_i^*(\zeta))^{p(s-t)/(2s)}
\left(\int_{\Gamma(\zeta)\setminus\overline{B(0,\delta)}}|\mathcal{R}f_i(w)|^t(1-|w|)^{t-\eta}\,d\tau(w)\right)^{p/(2s)}\,d\sigma(\zeta)\nonumber\\
&~&\qquad\lesssim\e^{p/2}\left(\int_{\mathbb{S}_n}(f_i^*(\zeta))^{p}\,d\sigma(\zeta)\right)^{(s-t)/2s}\times\nonumber\\
&~&\qquad\qquad\times
\left(\int_{\mathbb{S}_n}\left(\int_{\Gamma(\zeta)\setminus\overline{B(0,\delta)}}|\mathcal{R}f_i(w)|^t(1-|w|)^{t-\eta}\,d\tau(w)\right)^{p/(s+t)}
\,d\sigma(\zeta)\right)^{(s+t)/2s}\nonumber\\
&~&\qquad\lesssim\e^{p/2}\|f_i\|_{H^p}^{p(s-t)/2s}
\|f_i\|_{HT_{t,\alpha}^{(pt)/(s+t)}}^{pt/(2s)},\nonumber
\end{eqnarray}
where $\alpha=t-\eta-n-1$.
Choose a real number $\lambda$ such that
$$
\left(\frac{s+t}{pt}+\frac{t-\eta}{nt}\right)^{-1}<\lambda<\frac{pt}{s+t}<p.
$$
Let
$$
\beta=\frac{n\lambda}{pt/(s+t)}+\frac{\lambda(n+1+\alpha)}{t}-n-1
=n\lambda\left[\frac{s+t}{pt}+\frac{t-\eta}{nt}\right]-n-1.
$$
By Lemma 3.1 in \cite{WZ2} and Theorem 2.16 in \cite{KZ2005}
we know that
$$
\|f_i\|_{HT_{t,\alpha}^{(pt)/(s+t)}}
\lesssim\|\mathcal{R}f_i\|_{A^{\lambda}_{\beta+\lambda}}
\asymp\|f_i\|_{\mathrm{A}^{\lambda}_{\beta}}.
$$
From our choice of $\lambda$ we can easily see that $\beta>-1$.
Since $\lambda<p$, we get
$$
H^p\subseteq \mathrm{A}^p_{\beta}\subseteq \mathrm{A}^{\lambda}_{\beta},
$$
and the embedding is continuous. 
Thus $\|f_i\|_{\mathrm{A}^{\lambda}_{\beta}}\lesssim\|f_i\|_{H^p}$, and so
$$
\|f_i\|_{HT_{t,\alpha}^{(pt)/(s+t)}}
\lesssim\|f_i\|_{H^p}.
$$
Combining with (\ref{c5}) we get
\begin{equation}\label{c6}
\int_{\mathbb{S}_n}[(f_i^*(\zeta))^{(s-t)/s}I_2^{1/s}]^{p/2}\,d\sigma(\zeta)
\lesssim\e^{p/2}\|f_i\|_{H^p}^{p(s-t)/2s}\|f_i\|_{H^p}^{pt/(2s)}
=\e^{p/2}\|f_i\|_{H^p}^{p/2}.
\end{equation}
For $I_1$, since $\{\mathcal{R}f_i\}$ converges to 0 uniformly on compact subsets of $\B$,
by (\ref{c2}), there is a positive integer $N$ such that, for any $i\ge N$, we have
$$
I_1\le \e^s.
$$
Therefore, by H\"older's inequality
\begin{eqnarray}\label{c7}
&~&\int_{\mathbb{S}_n}[(f_i^*(\zeta))^{(s-t)/s}I_1^{1/s}]^{p/2}\,d\sigma(\zeta)\\
&~&\qquad=\int_{\mathbb{S}_n}(f_i^*(\zeta))^{p(s-t)/(2s)}I_1^{p/(2s)}\,d\sigma(\zeta)\nonumber\\
&~&\qquad\lesssim\left(\int_{\mathbb{S}_n}(f_i^*(\zeta))^{p}\,d\sigma(\zeta)\right)^{(s-t)/2s}
\left(\int_{\mathbb{S}_n}I_1^{p/(s+t)}\,d\sigma(\zeta)\right)^{(s+t)/2s}\nonumber\\
&~&\qquad\lesssim\e^{p/2}\|f_i\|_{H^p}^{p(s-t)/2s}.\nonumber
\end{eqnarray}
Combining (\ref{c4}), (\ref{c6}) and (\ref{c7}) we get
\begin{equation}\label{c8}
\lim_{i\to\infty}\|A^k_{\mu, s}(f_i)\|_{L^{p/2}(\mathbb{S}_n)}=0.
\end{equation}
We want to show that
\begin{equation}\label{c9}
\lim_{i\to\infty}\|\A_{\mu,s}^k(f_i)\|_{L^q(\mathbb{S}_n)}=0.
\end{equation}
Suppose (\ref{c9}) fails. Then there is a number $\gamma>0$ such that
$\{f_i\}$ has a subsequence $\{f_{i_j}\}$ such that
\begin{equation}\label{c10}
\|\A_{\mu,s}^k(f_{i_j})\|_{L^q(\mathbb{S}_n)}\ge\gamma
\end{equation}
for all $j=1,2,3...$. Since $\{f_i\}$ is a bounded sequence and 
$\A_{\mu,s}^k:\,H^p\rightarrow L^q(\mathbb{S}_n)$ is compact,
$\{f_{i_j}\}$ further has a subsequence $\{f_{i_{j_m}}\}$
such that
\begin{equation}\label{c11}
\lim_{m\to\infty}\|\A_{\mu,s}^k(f_{i_{j_m}})-g\|_{L^q(\mathbb{S}_n)}=0
\end{equation}
for some $g\in L^q(\mathbb{S}_n)$.
Since $0<p/2<p\le q<\infty$, we have $L^q(\mathbb{S}_n)\subseteq L^{p/2}(\mathbb{S}_n)$
and the embedding is continuous. Thus, by (\ref{c11}),
\begin{equation}\label{c12}
\lim_{i\to\infty}\|\A_{\mu,s}^k(f_{i_{j_m}})-g\|_{L^{p/2}(\mathbb{S}_n)}=0.
\end{equation}
Combining (\ref{c8}) and (\ref{c12}) we get
$$
\|g\|_{L^{p/2}(\mathbb{S}_n)}
\lesssim\|g-\A_{\mu,s}^k(f_{i_{j_m}})\|_{L^{p/2}(\mathbb{S}_n)}
+\|\A_{\mu,s}^k(f_{i_{j_m}})\|_{L^{p/2}(\mathbb{S}_n)}\to0
$$
as $m\to\infty$. Thus $g(\zeta)=0$ a.e. on $\mathbb{S}_n$.
Therefore $\|g\|_{L^q(\mathbb{S}_n)}=0$. Combining with (\ref{c11})
we get
$$
\lim_{m\to\infty}\|\A_{\mu,s}^k(f_{i_{j_m}})\|_{L^q(\mathbb{S}_n)}=0,
$$
which contradicts (\ref{c10}). Hence, we must have (\ref{c9}),
or (ii) is true.

(ii)$\Rightarrow$(i).
Suppose (ii) holds. 
Let $\{h_i\}$ be a bounded sequence in $H^p$. 
Then $\{h_i\}$ is a normal family in $\B$, and so it has a subsequence $\{h_{i_j}\}$
such that $\{h_{i_j}\}$ converges uniformly on compact subsets of $\B$ to a holomorphic function $h$. 
By Fatou's Lemma we know that $h\in H^p$.
Since $\A_{\mu,s}^k:H^p\rightarrow L^q(\mathbb{S}_n)$ is bounded, we know that
$\A_{\mu,s}^k(h_i)\in L^q(\mathbb{S}_n)$, and $\A_{\mu,s}^k(h)\in L^q(\mathbb{S}_n)$.

Let $g_j=h_{i_j}-h$. 
Then clearly $\{g_j\}$ is a also bounded sequence in $H^p$, 
and it converges to 0 uniformly on compact subsets of $\B$.
By (ii), we know that 
\begin{equation}\label{g0}
\lim_{j\to\infty}\|\A_{\mu,s}^k(g_j)\|_{L^q(\mathbb{S}_n)}=0.
\end{equation}

We first consider the case $s\ge1$. 
By Minkowski's inequality with index $s\ge 1$ we get that
$$
0\le |\A^k_{\mu,s}(h_{i_j})-\A^k_{\mu,s}(h)|\le \A^k_{\mu,s}(g_{j}).
$$
Thus
$$
\lim_{j\to\infty}\|\A^k_{\mu,s}(h_{i_j})-\A^k_{\mu,s}(h)\|_{L^q(\mathbb{S}_n)}
=\lim_{j\to\infty}\|\A^k_{\mu,s}(g_{j})\|_{L^q(\mathbb{S}_n)}=0.
$$
Hence, $\{\A^k_{\mu,s}(h_{i})\}$ has a convergent subsequence in $L^q(\mathbb{S}_n)$,
and so  $\A^k_{\mu,s}:\,H^p\to L^q(\mathbb{S}_n)$ is compact.

Next, we consider the case $0<s<1$.
In this case we obviously have
$$
|(\A^k_{\mu,s}(h_{i_j}))^s-(\A^k_{\mu,s}(h))^s|\le (\A^k_{\mu,s}(g_{j}))^s.
$$
Thus
$$
\|(\A^k_{\mu,s}(h_{i_j}))^s-(\A^k_{\mu,s}(h))^s\|^{q/s}_{L^{q/s}(\mathbb{S}_n)}
\le\|\A^k_{\mu,s}(g_{j})\|^q_{L^{q}(\mathbb{S}_n)}.
$$
Therefore, by (\ref{g0}),
\begin{equation}\label{lim1}
\lim_{j\to\infty}\|(\A^k_{\mu,s}(h_{i_j}))^s-(\A^k_{\mu,s}(h))^s\|_{L^{q/s}(\mathbb{S}_n)}=0.
\end{equation}
By Lemma~\ref{rudin} and Remark~\ref{rudin1},
we know that $\{h_{i_j}\}$ has a subsequence of $\{h_{i_{j_m}}\}$ such that 
\begin{equation}\label{limae}
\A^k_{\mu,s}(h_{i_{j_m}})\to \A^k_{\mu,s}(h)\quad \textrm{a.e. on }\mathbb{S}_n.
\end{equation}
Correspondingly, we denote by $g_{m}=h_{i_{j_m}}-h$.
Now we consider two cases. 
First, let $q\ge s$. Since $0<s<1$ and $q/s\ge1$, 
by Minkowski's inequality with index $q/s$ we get
$$
\|(\A^k_{\mu,s}(h_{i_{j_m}}))^s\|_{L^{q/s}(\mathbb{S}_n)}
\le\|(\A^k_{\mu,s}(h))^s\|_{L^{q/s}(\mathbb{S}_n)}+\|(\A^k_{\mu,s}(g_{m}))^s\|_{L^{q/s}(\mathbb{S}_n)},
$$
and similarly,
$$
\|(\A^k_{\mu,s}(h))^s\|_{L^{q/s}(\mathbb{S}_n)}
\le\|(\A^k_{\mu,s}(h_{i_{j_m}})^s\|_{L^{q/s}(\mathbb{S}_n)}+\|(\A^k_{\mu,s}(g_{m}))^s\|_{L^{q/s}(\mathbb{S}_n)}.
$$
By taking limits with $m\to\infty$ on both sides of the above inequalities,
and applying (\ref{g0}), we get
\begin{equation}\label{ak1}
\lim_{i\to\infty}\|\A^k_{\mu,s}(h_{i_{j_m}})\|^s_{L^q(\mathbb{S}_n)}
=\lim_{i\to\infty}\|(\A^k_{\mu,s}(h_{i_{j_m}}))^s\|_{L^{q/s}(\mathbb{S}_n)}
=\|(\A^k_{\mu,s}(h))^s\|_{L^{q/s}(\mathbb{S}_n)}
=\|(\A^k_{\mu,s}(h))\|^s_{L^{q}(\mathbb{S}_n)}.
\end{equation}
Next, let $0<q<s$. Since $0<q/s<1$ we have
Since
\begin{eqnarray*}
[\A^k_{\mu,s}(h_{i_{j_m}})(\zeta)]^q
&=&|\A^k_{\mu,s}((h_{i_{j_m}})(\zeta))^s|^{q/s}\\
&=&|(\A^k_{\mu,s}(h_{i_{j_m}})(\zeta))^s-(\A^k_{\mu,s}(h)(\zeta))^s+(\A^k_{\mu,s}(h)(\zeta))^s|^{q/s}\\
&\le&|(\A^k_{\mu,s}(h_{i_{j_m}})(\zeta))^s-(\A^k_{\mu,s}(h)(\zeta))^s|^{q/s}+(\A^k_{\mu,s}(h)(\zeta))^q,
\end{eqnarray*}
we get that, by (\ref{lim1}),
\begin{equation}\label{a1}
\|\A^k_{\mu,s}(h_{i_{j_m}})\|^q_{L^{q}(\mathbb{S}_n)}
\le\|(\A^k_{\mu,s}(h_{i_{j_m}}))^s-(\A^k_{\mu,s}(h))^s\|^{q/s}_{L^{q/s}(\mathbb{S}_n)}
+\|\A^k_{\mu,s}(h)\|^q_{L^{q}(\mathbb{S}_n)}
\end{equation}
Similarly, we can get
\begin{equation}\label{a2}
\|\A^k_{\mu,s}(h)\|^q_{L^{q}(\mathbb{S}_n)}
\le\|(\A^k_{\mu,s}(h))^s-(\A^k_{\mu,s}(h_{i_{j_m}}))^s\|^{q/s}_{L^{q/s}(\mathbb{S}_n)}
+\|\A^k_{\mu,s}(h_{i_{j_m}})\|^q_{L^{q}(\mathbb{S}_n)}
\end{equation}
Taking limit with $m\to\infty$ on both sides of (\ref{a1}) and (\ref{a2}), and using (\ref{lim1}),
we again obtain that
\begin{equation}\label{ak2}
\lim_{i\to\infty}\|\A^k_{\mu,s}(h_{i_{j_m}})\|^q_{L^{q}(\mathbb{S}_n)}
=\|\A^k_{\mu,s}(h)\|^q_{L^{q}(\mathbb{S}_n)}.
\end{equation}
Combining (\ref{ak1}), (\ref{ak2}) and (\ref{limae}), by Lemma~\ref{duren}, we get that
$$
\lim_{i\to\infty}
\|\A^k_{\mu,s}(h_{i_{j_m}})-\A^k_{\mu,s}(h)\|_{L^{q}(\mathbb{S}_n)}=0.
$$
Thus $\A^k_{\mu,s}:\,H^p\to L^q(\mathbb{S}_n)$ is compact, or (i) is true. The proof is complete.
\end{proof}

\begin{remark}
The proof for (ii)$\Rightarrow$(i) works for any $0<p,q,s<\infty$.
\end{remark}

\begin{thm}\label{van1}
Suppose $0<p<q<\infty$ or $p=q, s\geq 2$. Let  $\mu$ be a positive Borel measure on $\mathbb{B}_n$, finite
on compact sets of $\mathbb{B}_n$, and let $0<r<1$. 
Then for $k \in \mathbb N$,
$\A_{\mu,s}^k:H^p\rightarrow L^q(\mathbb{S}_n)$
is compact if and only if $\mu$ is a vanishing $(ks/n+1+s/p-s/q)$-Carleson measure,
that is
\begr\label{van0}
\lim_{|a|\rightarrow 1}\frac{\mu(D(a,r))}{(1-|a|)^{ks+(1+s/p-s/q)n}}=0.
\endr
\end{thm}

\begin{proof}
First, we consider the necessity. 
Let $\A_{\mu,s}^k:H^p\rightarrow L^q(\mathbb{S}_n)$ be compact.
Set
$$
f_{a,p}(z)=\frac{(1-|a|^2)^{mn}}{(1-\langle a,z\rangle)^{(1/p+m)n}}.
$$
It can be easily checked that $f_{a,p}\in H^{p}$,
$\sup_{a\in\B}\|f_{a,p}\|_{H^{p}}\lesssim 1$, 
and $f_{a,p}$ converges uniformly to zero on compact subsets of $\B$ as $|a|\to1$.
Since $\A_{\,\mu,\,s}^k:H^p\rightarrow L^q(\mathbb{S}_n)$
is compact, by Lemma~\ref{compact}, we get that
\[
\lim_{|a|\to1}\|\A_{\,\mu,\,s}^k(f_{a,p})\|_{L^q(\mathbb{S}_n)}=0.
\]
Therefore, by (\ref{t5.42}) we have
\[
\lim_{|a|\to 1}\frac{\mu(D(a,r))}{(1-|a|)^{ks+(1+s/p-s/q)n}}=0.
\]

Conversely, we first consider the case $0<p<q<\infty$.
Assume that (\ref{van0}) holds. 
Then obviously (\ref{t5.41}) holds and so, by Theorem~\ref{thm1.1a},
$\A_{\,\mu,\,s}^k:H^p\rightarrow L^q(\mathbb{S}_n)$ is bounded.
Let $\{g_i\}_{i=1}^\infty$ be a bounded sequence in $H^p$
and converges to 0 uniformly on compact subsets of $\B$.
By (\ref{van0}), for any $\varepsilon>0$, there is a $\delta\in(0,1)$
such that
\begin{equation}\label{ep}
\frac{\mu(D(a,r))}{(1-|a|)^{ks+(1+s/p-s/q)n}}<\varepsilon^{s/p}
\end{equation}
for any $a\in\B\setminus\overline{B(0,\delta)}$.
Let $\alpha=ns\left(\frac{1}{p}-\frac{1}{q}\right)-n-1$.
By (\ref{11}) we have
\begr\label{11-2}
&~&(\A^k_{\mu, s}g_i(\zeta))^s\\
&~&\qquad\lesssim\sum_{a_j \in \widetilde{\Gamma}(\zeta)}\int_{D(a_j, 2r)}|g_i(w)|^s(1-|w|)^{\alpha}\,dV(w)
\frac{\mu(D(a_j,r))}{(1-|a_j|)^{ks+n(1+\frac{s}{p}-\frac{s}{q})}}\nonumber\\
&~&\qquad\lesssim\sum_{a_j \in \widetilde{\Gamma}(\zeta)\cap \overline{B(0,\delta)}}
\int_{D(a_j,2r)}|g_i(w)|^s(1-|w|)^{\alpha}\,dV(w)\frac{\mu(D(a_j, r))}{(1-|a_j|)^{ks+n(1+\frac{s}{p}-\frac{s}{q})}}\nonumber\\
&~&\qquad\quad+\sum_{a_j \in \widetilde{\Gamma}(\zeta)\setminus\overline{B(0,\delta)}}
\int_{D(a_j,2r)}|g_i(w)|^s(1-|w|)^{\alpha}\,dV(w)\frac{\mu(D(a_j, r))}{(1-|a_j|)^{ks+n(1+\frac{s}{p}-\frac{s}{q})}}\nonumber\\
&~&\qquad=I_1+I_2.\nonumber
\endr
By Lemma~\ref{l51} and (\ref{ep}),
\begin{eqnarray*}
I_2\lesssim \varepsilon^{s/p}\int_{\widetilde{\widetilde{\Gamma}}(\zeta)}|g_i(w)|^s(1-|w|)^{\alpha}\,dV(w)
\end{eqnarray*}
Thus, by (\ref{h-ht}),
\begin{eqnarray}\label{i2}
\int_{\mathbb{S}_n}I_2^{q/s}\,d\sigma(\zeta)
&\lesssim&\varepsilon\int_{\mathbb{S}_n}\left(\int_{\widetilde{\widetilde{\Gamma}}(\zeta)}
|g_i(w)|^s(1-|w|)^{\alpha}\,dV(w)\right)^{q/s}\,d\sigma(\zeta)\\
&\lesssim&\varepsilon\,\|g_i\|_{HT_{s,\alpha}^q}^q
\lesssim\varepsilon\,\|g_i\|_{H^p}^q
\lesssim\varepsilon.\nonumber
\end{eqnarray}
For $I_1$, since $w\in D(a_j,2r)$, and $a_j\in\widetilde{\Gamma}(\zeta)\cap\overline{B(0,\delta)}$,
there is a $\delta'\in(0,1)$ such that $w\in\overline{B(0,\delta')}$.
Since $g_i\to0$ uniformly on compact subsets of $\B$, by (\ref{t5.41}) we get that
$$
I_1\lesssim\int_{\overline{B(0,\delta')}}|g_i(w)|^s(1-|w|)^{\alpha}\,dV(w)\to 0
$$
as $i\to\infty$. Thus there is a positive integer $N$ such that
\begin{equation}\label{i1}
\int_{\mathbb{S}_n}I_1^{q/s}\,d\sigma(\zeta)\lesssim\varepsilon
\end{equation}
for any $i\ge N$. Combining (11-2), (\ref{i2}) and (\ref{i1}) we get that
$$
\|\A_{\mu,\,q}^k(g_i)\|_{L^q(\mathbb{S}_n)}^q
\lesssim\int_{\mathbb{S}_n}(\A_{\mu,\,q}^k(g_i)(\zeta))^{q/s}\,d\sigma(\zeta)
\lesssim\int_{\mathbb{S}_n}(I_1^{q/s}+I_2^{q/s})\,d\sigma(\zeta)
\lesssim\varepsilon
$$
for any $i\ge N$. Thus
\begin{equation*}
\lim_{i\to\infty}\|\A_{\mu,\,q}^k(g_i)\|_{L^q(\mathbb{S}_n)}^q=0.
\end{equation*}
By Lemma~\ref{compact}, 
$\A_{\mu,s}^k:H^p\rightarrow L^q(\mathbb{S}_n)$ is compact.

Next, we consider the case $p=q, s\geq 2$,
From \eqref{Rin}, 
by dividing 
$\widetilde{\Gamma}(\zeta)$ into $\widetilde{\Gamma}(\zeta)\cap\overline{B(0,\delta)}$
and 
$\widetilde{\Gamma}(\zeta)\setminus\overline{B(0,\delta)}$
and arguing as in the previous case we can also see that 
$$
\lim_{i\to\infty}\|\A_{\mu,\,q}^k(g_i)\|_{L^q(\mathbb{S}_n)}^q=0.
$$
We omit the details here.
Thus, by Lemma~\ref{compact},
$\A^k_{\,\mu,s}: H^p\rightarrow L^q(\mathbb{S}_n)$ is compact.
The proof is complete.
\end{proof}

For other cases, we need some lemmas.

\begin{lem}\label{tent01}
Let $0<p,s<\infty$, let $k$ be a nonnegative integer, and let $\mu$ be a finite, positive Borel measure on $\B$. 
If a family $\mathcal{G}$ of functions in $T^p_s(\mu)$ is relatively compact then
for any $\e>0$, there exists an $r\in(0,1)$ such that 
$$
\sup_{g\in\mathcal{G}}
\int_{\mathbb{S}_n}\left(\int_{\Gamma(\zeta)\setminus{\overline{B(0,r)}})}|g(z)|^s\,d\mu(z)\right)^{p/s}\,d\sigma(\zeta)<\e.
$$ 
\end{lem}

\begin{proof} Since $\mathcal{G}$ is a relatively compact set in $T^p_s(\mu)$,
for any $\e>0$, we can find a finite $(\e/2)^{1/p}$-net $\{g_j\}_{j=1}^l$ for the family $\mathcal{G}$.
Let $g$ be any element in $\mathcal{G}$.
Then there is a $j\in\{1,...,l\}$ such that
$$
\|g-g_j\|_{T^p_s(\mu)}^p<\e/2.
$$
By the dominated convergence theorem, 
for the above $\e>0$, and
for each $j\in\{1,...,l\}$ there is an $r_j\in(0,1)$ such that
$$
\int_{\mathbb{S}_n}
\left(\int_{\Gamma(\zeta)\setminus{\overline{B(0,r_j)}})}|g_j(z)|^s\,d\mu(z)\right)^{p/s}\,d\sigma(\zeta)<\frac{\e}2.
$$
Let $r=\max\{r_j:\,j=1,...,l\}$. 
Then
\begin{eqnarray*}
&~&\int_{\mathbb{S}_n}\left(\int_{\Gamma(\zeta)\setminus{\overline{B(0,r)}})}|g(z)|^s\,d\mu(z)\right)^{p/s}\,d\sigma(\zeta)\\
&~&\qquad\lesssim\int_{\mathbb{S}_n}\left(\int_{\Gamma(\zeta)\setminus{\overline{B(0,r)}})}|g(z)-g_j(z)|^s\,d\mu(z)\right)^{p/s}\,d\sigma(\zeta)\\
&~&\qquad\qquad+\int_{\mathbb{S}_n}\left(\int_{\Gamma(\zeta)\setminus{\overline{B(0,r)}})}|g_j(z)|^s\,d\mu(z)\right)^{p/s}\,d\sigma(\zeta)\\
&~&\qquad<\frac{\e}{2}+\frac{\e}{2}=\e.
\end{eqnarray*}
The proof is complete.
\end{proof}

\begin{cor}\label{tent02}
Let $0<p,q,s<\infty$, let $k$ be a nonnegative integer, and let $\mu$ be a finite, positive Borel measure on $\B$. 
Let $\mathcal{A}_{\mu,s}^k:\,H^p\to L^q(\mathbb{S}_n)$ be compact. 
Let $\mathcal F$ be a bounded family in $H^p$.
Then, for any $\e>0$, there exists an $r\in(0,1)$ such that  
$$ 
\sup_{f\in\mathcal{F}}
\int_{\mathbb{S}_n}\left(\int_{\Gamma(\zeta)\setminus{\overline{B(0,r)}})}
|R^kf(z)|^s\,\frac{d\mu(z)}{(1-|z|)^n}\right)^{q/s}\,d\sigma(\zeta)<\e.
$$
\end{cor}

\begin{proof}
Define an operator $A^k_s$ by
$$
A^k_s(z)=\frac{R^kf(z)\,\chi_{\Gamma(\zeta)}(z)}{(1-|z|)^{n/s}}.
$$
Obviously, for any $f\in H^p$,
$\|\mathcal{A}_{\mu,s}^k(f)\|_{L^q(\mathbb{S}_n)}=\|A^k_s(f)\|_{T^q_s(\mu)}$.
Thus $\mathcal{A}_{\mu,s}^k:\,H^p\to L^q(\mathbb{S}_n)$ is compact
if and only if 
$A^k_s:\,H^p\to T^q_s(\mu)$ is compact. 
Therefore, if $\mathcal{F}$ is a bounded family in $H^p$ then
$A^k_s(\mathcal{F})$ is a relatively compact set in $T^p_s(\mu)$, 
and the result then follows form Lemma~\ref{tent01}.
\end{proof}

\begin{lem}\label{test-c}
Let $k$ be a positive integer, let $0<p, q ,s<\infty$ and let $0<r<1$. 
Let $\mu$ be a finite, positive Borel measure on $\mathbb{B}_n$ on $\mathbb{B}_n$.
Let $R\in(0,1)$ 
Let $\Phi_{\mu}(z)=\mu(D(z,r))/(1-|z|)^{n+ks}$, and let 
$(\Phi_{\mu})_R=\Phi_{\mu}\cdot\chi_{\B\setminus\overline{B(0,R)}}$.
If $\mathcal{A}_{\mu, s}^k: H^p \rightarrow L^q\left(\mathbb{S}_n\right)$ is compact then
for any $r$-lattice $Z=\{z_j\}\in\B$ and any complex sequence $\lambda=\{\lambda_j\}\in T^p_2(Z)$,
we have
\begin{equation}\label{lambda-phi-c}
\lim_{R\to 1}\int_{\mathbb{S}_n}\left(\sum_{z_j\in \Gamma(\zeta)}|\lambda_j|^s(\Phi_{\mu,2r})_R(z_j)\right)^{q/s}\,d\sigma(\zeta)
=0.
\end{equation}
\end{lem}

\begin{proof}
Suppose that $\A^k_{\,\mu,s}: H^p\to L^q(\mathbb{S}_n)$ is compact. 
Let $0<r<1$, and let $Z=\{z_j\}$ be an $r$-lattice in $\B$ such that $z_j\neq 0$ for all $j$. 
Denote by $\lambda=\{\lambda_j\}$ a sequence complex numbers.
Let $B_{T_2^p(Z)}=\{\lambda\in T_2^p(Z):\|\lambda\|_{T_2^p(Z)}\le1\}$.
For each $0\leq R<1$ and $t$ with $t>\max(2/p,1)$,
consider the operator
$$
S_{Z,R}(\lambda)(z)=\sum_{|z_j|\ge R}\lambda_j\left(\frac{1-|z_j|}{1-\langle z_j,z\rangle}\right)^{nt},\quad z\in\B.
$$
We write $S_{Z,0}(\lambda)=S_{Z}(\lambda)$. 
By Lemma \ref{l4.1}, $S_{Z,R}:T^p_2(Z)\to H^p$ is bounded, and
$$
\|S_{Z,R}(\lambda)\|_{H^p}\le C\|\lambda\|_{T_2^p(Z)}, \quad \textrm{for each } R\in[0,1).
$$
Since $\A^k_{\,\mu,s}: H^p\to L^q(\mathbb{S}_n)$ is compact, by Corollary~\ref{tent02},
for any $\ve>0$, there exists an $\tau\in(0,1)$ such that  
\begin{equation}\label{szr}
\sup_{\lambda\in B_{T^p_2(Z)}}
\int_{\mathbb{S}_n}\left(\int_{\Gamma(\zeta)\setminus{\overline{B(0,\tau)}})}
|\mathcal{R}^k(S_{Z,R}(\lambda))(z)|^s\,\frac{d\mu(z)}{(1-|z|)^n}\right)^{q/s}\,d\sigma(\zeta)<\ve^q
\end{equation}
for any $R\in[0,1)$.

Since $\{z_j\}$ is separated and $t>1$, there is an $R_0$ such that, for any $R\in[R_0,1)$,
we have $\sum_{|z_j|\ge R}(1-|z_j|)^{nt}<\ve^2$. 
Thus, by H\'older's inequality and Lemma~\ref{arsen} we get that, for $|z|\le\tau$ and $R\ge R_0$,
\begin{eqnarray*}
|\mathcal{R}^k\circ S_{Z,R}(\{\lambda)(z)|
&\lesssim&\sum_{|z_j|\ge R}|\lambda_j|(1-|z_j|)^{nt}\\
&\lesssim&\left(\sum_{|z_j|\ge R}|\lambda_j|^2(1-|z_j|)^{nt}\right)^{1/2}\cdot\ve\\
&\lesssim&\inf_{\zeta\in\mathbb{S}_n}\left(\sum_{|z_j|\ge R}|\lambda_j|^2\frac{(1-|z_j|}{|1-\langle z_j,\zeta\rangle|}\right)^{1/2}\cdot \ve\\
&\lesssim&\left(\int_{\mathbb{S}_n}\left(\sum_{|z_j|\ge R}|\lambda_j|^2\frac{(1-|z_j|}{|1-\langle z_j,\zeta\rangle|}\right)^{p/2}\right)^{1/p}\cdot \ve\\
&\lesssim&\|\lambda\|_{T^p_2(Z)}\cdot\ve. 
\end{eqnarray*} 
Combining with (\ref{szr}) we get
\begin{equation}\label{szr1}
\int_{\mathbb{S}_n}\left(\int_{\Gamma(\zeta)}
|\mathcal{R}^k(S_{Z,R}(\lambda))(z)|^s\,\frac{d\mu(z)}{(1-|z|)^n}\right)^{q/s}\,d\sigma(\zeta)
\lesssim\|\lambda\|_{T^p_2(Z)}^q\cdot\ve^q
\end{equation}
for all $R\ge R_0$ and $\lambda\in T^p_2(Z)$.
That is
\begin{equation}\label{szr2}
\int_{\mathbb S_n}\left(\int_{\Gamma(\zeta)}\left|\sum_{|z_j|\geq R}\lambda_j
\frac{(1-|z_j|)^{nt}}{(1-\langle z_j,z\rangle)^{nt+k}}\right|^s
\frac{d\mu(z)}{(1-|z|)^n}\right)^{{q/s}}d\sigma(\zeta)
\lesssim \|\lambda\|_{T^p_2(Z)}^q\cdot\ve^q
\end{equation}
for all $R\in[R_0,1)$ and $\lambda\in T^p_2(Z)$.
Following the proof of Lemma~\ref{test-b} we can get
(\ref{lambda-phi-c}). 
We omit the details here.
\end{proof}

\begin{thm}\label{van2}
Let $0<p, q ,s<\infty$, let $k \in \mathbb N$, let $0<r<1$ and 
let $\mu$ be a positive Borel measure on $\mathbb{B}_n$, finite on compact subsets of $\B$. 
Let $\Phi_{\mu}(z)=\mu(D(z,r))/(1-|z|)^{n+ks}$,
and let $(\Phi_{\mu})_R=\Phi_{\mu}\cdot\chi_{\B\setminus\overline{B(0,R)}}$.
Then the following statements hold.
\begin{itemize}
\item[(i)]If $p=q$ and $s<2$, then
$\mathcal{A}^k_{\mu, s} : H^p \rightarrow L^q(\mathbb S_n)$ is compact if and only if 
$$
\lim_{R\to1}\|(\Phi_{\mu})_R\|_{T^{\infty}_{2/(2-s)}(\tau)}=0.
$$
\item[(ii)] If $p>q$ and $s\ge2$, then 
$\mathcal{A}_{\mu, s}^k: H^p \rightarrow L^q(\mathbb{S}_n)$ is compact if and only if 
$$
\lim_{R\to 1}\|(\Phi_{\mu})_R\|_{T^{pq/(s(p-q)}_{\infty}(\tau)}=0.
$$
\item[(iii)] If $p>q$ and $s<2$, then $\mathcal{A}_{\mu, s}^k: H^p \to L^q(\mathbb{S}_n)$ 
is compact if and only if $\Phi_{\mu} \in T_{2 /(2-s)}^{pq/(s(p-q))}(\tau)$.
\end{itemize}
\end{thm}

\begin{proof}
(i). We first consider the sufficiency. Suppose that
$$
\lim_{R\to1}\|(\Phi_{\mu})_R\|_{T^{\infty}_{2/(2-s)}(\tau)}=0.
$$
By Remark~\ref{r-thm1.1b},
for any $r$-lattice $Z=\{a_j\}$ in $\B$,
$$
\lim_{R\to1}\|(\Phi_{\mu})_R(a_j)\|_{T^{\infty}_{2/(2-s)}(Z)}=0.
$$
Thus, for any $\ve>0$, there is an $R_0\in(0,1)$ such that
\begin{align}\label{14}
\|(\Phi_{\mu})_R(a_j)\|_{T_{2/(2-s)}^\infty(Z)}
=\sup_{a\in\B}\frac{1}{(1-|a|)^n}\sum_{a_j\in\Lambda(a)}
\left|(\Phi_{\mu})_R(a_j)\right|^{2/(2-s)}(1-|a_j|)^n<\ve^{s}
\end{align}
for any $R\in[R_0,1)$.
Fix such $R$.
Let $\{f_i\}$ be a bounded sequence in $H^p$. 
Then it has a subsequence $\{f_{i_m}\}$ that converges uniformly on compact subsets of $\mathbb{B}_n$ 
to a function $f\in H^p$. 
Let $g_m=f_{i_m}-f$.
Then $\{\mathcal{R}^kg_{m}\}$ converges to 0 uniformly on compact subsets in $\B$. 
Choose $m_0$ such that $|\mathcal{R}^kg_m(z)|<\ve$ for all $m\ge m_0$ and all $z\in \overline{B(0,R)}$.
By \eqref{12}, we obtain
\begin{eqnarray*}
(\mathcal{A}^k_{\mu, s}(g_{m})(\zeta))^s
&\lesssim&\sum_{a_j\in \widetilde{\Gamma}(\zeta)\cap \overline{B(0,R)}}
\left(\int_{D(a_j,2r)}|\mathcal{R} g_{m}(w)|^2(1-|w|)^2\,d\tau(w)\right)^{s/2}\cdot\Phi_{\mu}(a_j)\\
 &~&\qquad+\sum_{a_j \in \widetilde{\Gamma}(\zeta)\setminus\overline{B(0,R)}}
 \left(\int_{D(a_j,2r)}|\mathcal{R} g_{m}(w)|^2(1-|w|)^2\,d\tau(w)\right)^{s/2}\cdot \Phi_{\mu}(a_j)\\
&=&I_1(\zeta)+I_2(\zeta).
\end{eqnarray*}
Let $\gamma=\{\gamma_j\}$, where
$$
\gamma_j=\left(\int_{D(a_j,2r)}|\mathcal{R} g_{m}(w)|^2(1-|w|)^2\chi_{\B\cap\overline{B(0,R)}}\,d\tau(w)\right)^{s/2}
$$
Then, for any $m\ge m_0$,
\begin{eqnarray*}
\|\gamma\|_{T^{p/s}_{2/s}(Z)}^{p/s}
&=&\int_{\mathbb{S}_n}\left(\sum_{a_j\in\Gamma(\zeta)\cap\overline{B(0,R)}}
\int_{D(a_j,2r)}|\mathcal{R}g_{m}(w)|^2(1-|w|)^2\,d\tau(w)\right)^{p/2}\,d\sigma(\zeta)\\
&\lesssim& \int_{\mathbb{S}_n}
\left(\int_{\widetilde{\Gamma}(\zeta)\cap\overline{B(0,R)}}|\mathcal{R}g_{m}(w)|^2(1-|w|)^2\,d\tau(w)\right)^{p/2}\,d\sigma(\zeta)\\
&\lesssim&\ve^{p}.
\end{eqnarray*}
Thus, similar to the proof of (\ref{lambda-phi1}), we get that
$$
\int_{\mathbb{S}_n}(I_1(\zeta))^{p/s}\,d\sigma(\zeta)
\lesssim\ve^p\|\Phi_{\mu}(a_j)\|_{T^{\infty}_{2/(2-s)}(Z)}^{p/s}.
$$
Again, similar to the proof of (\ref{lambda-phi1}), and by (\ref{14}), we get that
$$
\int_{\mathbb{S}_n}(I_2(\zeta))^{p/s}\,d\sigma(\zeta)
\lesssim \|g_{m}\|_{H^p}^p\|(\Phi_{\mu})_R(a_j)\|_{T^{\infty}_{2/(2-s)}(Z)}^{p/s}
\lesssim \|g_{m}\|_{H^p}^p\ve^p.
$$
Combining the above two inequalities, we get that, for all $m\ge m_0$,
\begin{eqnarray*}
\|\mathcal{A}^k_{\mu, s}(g_{m})\|^p_{L^p(\mathbb S_n)}
\lesssim\ve^p\|\Phi_{\mu}\|^{p/s}_{T^\infty_{2/(2-s)}(Z)}
   +\|g_{m}\|^p_{H^p}\ve^p
\lesssim\ve^p.
\end{eqnarray*}
Therefore,
\[
\lim_{m\to\infty}\|\A_{\,\mu,\,s}^k(g_{m})\|_{L^q(\mathbb{S}_n)}^s=0.
\]
Thus, by Lemma \ref{duren} we have
\[
\lim_{m\to\infty}\|\A_{\,\mu,\,s}^k(f_{i_m})-\A_{\mu,s}^k(f)\|_{L^q(\mathbb{S}_n)}=0,
\]
that is, $\A_{\,\mu,\,s}^k:\,H^p\rightarrow L^q(\mathbb{S}_n)$ is compact.

Conversely,  suppose that $\A^k_{\,\mu,s}: H^p\to L^q(\mathbb{S}_n)$ is compact. 
Let $0<r<1$, and let $Z=\{z_j\}$ be an $r$-lattice in $\B$ such that $z_j\neq 0$ for all $j$. 
By Lemma~\ref{test-c}, 
there is an $R_0\in (0,1)$ such that, for any $R\in[R_0,1)$, and
for any $\lambda=\{\lambda_j\}\in T^p_2(Z)$,
$$
\int_{\mathbb{S}_n}\left(\sum_{z_j\in \Gamma(\zeta)}|\lambda_j|^s(\Phi_{\mu,2r})_R(z_j)\right)^{q/s}\,d\sigma(\zeta)<\ve^q.
$$
Following the proof of (\ref{4}), we get that
for $\rho>1$ large enough, 
and for any $\varphi=\{\varphi_j\}\in T^1_{(2\rho/(2-s))'}$,
$$
\sum_j|\varphi_j|\,(\Phi_{\mu,2r})_R^{1/\rho}(z_j)(1-|z_j|)^n<\ve^{\rho}.
$$
for any $R\in[R_0,1)$.
By Lemma~\ref{dual} we obtain that 
$\{(\Phi_{\mu,2r})_R(z_j)\}\in T^{\infty}_{{2}/{(2-s)}}(Z)$, and
$$
\|\{(\Phi_{\mu,2r})_R(z_j)\}\|_{T^{\infty}_{{2}/{(2-s)}}(Z)}<\ve
$$
for any $R\in[R_0,1)$.  
By Lemma~\ref{carleson-v},
this is equivalent to that
$$
d\omega(z)=\sum_{j=1}^{\infty}|(\Phi_{\mu,2r}(z_j))|^{2/(2-s)}(1-|z_j|)^n\,d\delta_{z_j}(z)
$$
is a vanishing Carleson measure on $\B$.
Let $d\widetilde{\omega}(z)=|\Phi_{\mu}(z)|^{2/(2-s)}(1-|z|)^n\,d\tau(z)$.
By (\ref{carleson-cd}), we get that $d\widetilde{\omega}$
is a vanishing Carleson measure on $\B$.
By Lemma~\ref{carleson-v} again, this is equivalent to that
$$
\lim_{R\to1}\|(\Phi_{\mu})_R\|_{T^{\infty}_{2/(2-s)}(\tau)}=0,
$$
completing the proof of (i).

The proofs for (ii) and (iii) are similar. We only give a sketch of these proofs.
Suppose
$$
\lim_{R\to 1}\|(\Phi_{\mu})_R\|_{T^{pq/(s(p-q)}_{\infty}(\tau)}=0.
$$
Then, for a fixed $\ve>0$, there exists $R_0\in (0,1)$ such that
$$
\int_{\mathbb{S}_n}\left(\sup_{z\in\Gamma(\zeta)\setminus\overline{B(0,R)}}
\Phi_{\mu}(z)\right)^{pq/(s(p-q))}d\sigma(\zeta)<\ve^{p/(p-q)}
$$
for any $R\in[R_0,1)$.
Let $\{f_i\}$ be a bounded sequence in $H^p$. 
Then it has a subsequence $\{f_{i_m}\}$ that converges uniformly on compact subsets of $\mathbb{B}_n$ 
to a function $f\in H^p$. 
Let $g_{m}=f_{i_m}-f$.
Then $\{\mathcal{R}^kg_{m}\}$ converges to 0 uniformly on compact subsets in $\B$. 
Choose $m_0$ such that $|\mathcal{R}^kg_{m}(z)|<\ve$ for all $m\geq m_0$ and all $z\in\overline{B(0,R)}$.
By Lemma \ref{dmk}, Fubini's theorem, H\"{o}lder's inequality
and Lemma \ref{duren}  we obtain
\begr
\|\A_{\,\mu,\,s}^k(g_{m})\|_{L^q}^q
&\lesssim&\int_{\mathbb{S}_n}\left(\int_{\Gamma(\zeta)\cap \overline{B(0,R)}}
|g_m(z)|^s\Phi_{\mu}(z)\,d\tau(z)\right)^{q/s}\,d\sigma(\zeta)\nonumber\\
&~&\qquad+\int_{\mathbb{S}_n}\left(\int_{\Gamma(\zeta)\setminus\overline{B(0,R)}}
|g_{m}(z)|^s\Phi_{\mu}(z)\,d\tau(z)\right)^{q/s}\,d\sigma(\zeta)\nonumber\\
&\lesssim&\ve+\int_{\mathbb{S}_n}\left(g^*_m(\zeta)\right)^{(s-2)q/s}
\left(\int_{\Gamma(\zeta)\setminus\overline{B(0,R)}}
\Phi_{\mu}(z)\,d\tau(z)\right)^{q/s}d\sigma(\zeta)\nonumber\\
&\leq&\ve+\|g^*_m\|_{L^p}^q\left(\int_{\mathbb{S}_n}\left(\int_{\Gamma(\zeta)\setminus\overline{B(0,R)}}
\Phi_{\mu}(z)\,d\tau(z)\right)^{pq/(s(p-q))}\,d\sigma(\zeta)\right)^{(p-q)/p}\nonumber\\
&\lesssim&\ve\nonumber.
\endr
Therefore,
\[
\lim_{m\to\infty}\|\A_{\,\mu,\,s}^k(g_{m})\|_{L^q(\mathbb{S}_n)}^q=0.
\]
By arguing as in the previous case we see that $\A^k_{\,\mu,s}: H^p\rightarrow L^q(\mathbb{S}_n)$ is compact.

Conversely, suppose that $\A^k_{\,\mu,s}: H^p\to L^q(\mathbb{S}_n)$ is compact.
Let $0<r<1$, and let $Z=\{z_j\}$ be an $r$-lattice in $\B$ such that $z_j\neq 0$ for all $j$. 
By Lemma~\ref{test-c}, 
there is an $R_0\in (0,1)$ such that, for any $R\in[R_0,1)$, and
for any $\lambda=\{\lambda_j\}\in T^p_2(Z)$,
$$
\int_{\mathbb{S}_n}\left(\sum_{z_j\in \Gamma(\zeta)}|\lambda_j|^s(\Phi_{\mu,2r})_R(z_j)\right)^{q/s}\,d\sigma(\zeta)<\ve^q.
$$
Following the proof of (\ref{4}), we get that
for $\rho>1$ large enough, 
and for any $\varphi=\{\varphi_j\}\in T^{({pq\rho/(s(p-q))})'}_{\sigma}(Z)$,
$$
\sum_j|\varphi_j|\,(\Phi_{\mu,2r})_R^{1/\rho}(z_j)(1-|z_j|)^n<\ve^{\rho}.
$$
for any $R\in[R_0,1)$.
By Lemma~\ref{dual} we obtain that 
$\{\Phi_{\mu,2r}(z_j)\}\in T_{\infty}^{pq/(s(p-q))}(Z)$, and
$$
\|\{(\Phi_{\mu,2r})_R(z_j)\}\|_{T_{\infty}^{pq/(s(p-q))}(Z)}<\ve
$$
for any $R\in[R_0,1)$. 
By Lemma~\ref{3} we get that
$$
\lim_{R\to 1}\|(\Phi_{\mu})_R\|_{T^{pq/(s(p-q)}_{\infty}(\tau)}=0.
$$

(iii) Suppose that $\Phi_{\mu} \in T_{2 /(2-s)}^{pq/(s(p-q))}(\tau)$.
By the dominated convergence theorem, we get
$$
\lim_{R\to1}
\int_{\mathbb{S}_n}
\left(\int_{\Gamma(\zeta)\setminus\overline{B(0,R)}}
\left(\Phi_{\mu}(w)\right)^{2/(2-s)}\,d\tau(w)\right)^{((2-s)/2)pq/(s(p-q))}\,d\sigma(\zeta)=0.
$$
Thus, for any $\ve>0$, there is an $R_0\in(0,1)$ such that, for any $R\in[R_0,1)$,
$$
\int_{\mathbb{S}_n}
\left(\int_{\Gamma(\zeta)\setminus\overline{B(0,R)}}
\left(\Phi_{\mu}(w)\right)^{2/(2-s)}\,d\tau(w)\right)^{((2-s)/2)pq/(s(p-q))}\,d\sigma(\zeta)<\ve^{s(p-q)/(pq)}.
$$
Using this fact, similar to the proof of (ii), 
and following the proof of the sufficiency part of Theorem~\ref{thm1.1c} (ii),
we can get $\A^k_{\,\mu,s}: H^p\rightarrow L^q(\mathbb{S}_n)$ is compact. We omit the details.

The converse is obvious from Theorem~\ref{thm1.1c} (ii). 
The proof is complete.
\end{proof}

Again, by letting $s=q$ in Theorem~\ref{van1} and Theorem~\ref{van2},
and using (\ref{Aps}), we get the following result, which was first
proved on the unit disk by Pel\'{a}ez in \cite{AP2016}.

\begin{cor}\label{v-cm}
Let $0<p, q ,s<\infty$, let $k \in \mathbb N$, let $0<r<1$ and 
let $\mu$ be a positive Borel measure on $\mathbb{B}_n$, finite on compact subsets of $\B$. 
Let $\Phi_{\mu}(z)=\mu(D(z,r))/(1-|z|)^{n+kq}$,
and for $0<R<1$, let $(\Phi_{\mu})_R=\Phi_{\mu}\cdot\chi_{\B\setminus\overline{B(0,R)}}$.
Then the following statements hold.
\begin{itemize}
\item[(i)] If $p<q$ or $p=q\ge 2$ 
then $\mathcal{R}^k: H^p\to L^q(\mu)$  is compact if and only if 
$$
\lim_{|z|\to 1}\frac{\mu(D(z,r))}{(1-|z|)^{kq+nq/p}}=0.
$$
\item[(ii)] If $p=q<2$
then $\mathcal{R}^k: H^p\to L^q(\mu)$  is compact if and only if 
$$
\lim_{R\to1}\|(\Phi_{\mu})_R\|_{T^{\infty}_{2/(2-q)}(\tau)}=0.
$$
\item[(iii)] If $p>q\ge 2$ 
then $\mathcal{R}^k: H^p\to L^q(\mu)$  is compact if and only if 
$$
\lim_{R\to1}\|(\Phi_{\mu})_R\|_{T^{p/(p-q)}_{\infty}(\tau)}=0.
$$
\item[(iv)] If $p>q$ and $q<2$ 
then $\mathcal{R}^k: H^p\to L^q(\mu)$  is compact if and only if 
$
\Phi_{\mu}\in T^{p/(p-q)}_{2/(2-q)}(\tau).
$
\end{itemize}
\end{cor}

\end{document}